\documentclass[10pt]{article}
\usepackage[sectionbib]{natbib}
\usepackage{array,epsfig,fancyheadings,rotating}
\usepackage[driverfallback=dvipdfm]{hyperref}

\usepackage{amsmath}
\usepackage{amssymb}
\usepackage{amsfonts}
\usepackage{multirow}
\usepackage{amsthm}
\usepackage{mathtools}
\RequirePackage{bbm}
\usepackage{bm}
\usepackage{color}
\usepackage{algpseudocode,algorithm}
\usepackage{array}
\usepackage{eqparbox}

\newtheorem{theorem}{Theorem}
\newtheorem{lemma}{Lemma}

\newtheorem{proposition}{Proposition}
\newtheorem{assumption}{Assumption}
\theoremstyle{definition}

\newtheorem{remark}{Remark}
\renewcommand{\thefootnote}{\arabic{footnote}}

\newcommand{\argmin}{\operatornamewithlimits{arg\,min}}

\newcommand*{\rom}[1]{\expandafter\@slowromancap\romannumeral #1@}

\def \var{\text{var}}

\def \trace{\text{tr}}
\def \expect{\mathbb{E}}
\def \prob{\mathbb{P}}
\def \real{\mathbb{R}}
\def \diag{\text{diag}}
\def \ran{\mathcal{R}}
\def \({\left(}
\def \){\right)}
\def \[{\left[}
\def \]{\right]}
\DeclarePairedDelimiter\smallnorm{\lVert}{\rVert}
\providecommand{\norm}[1]{\left\lVert#1\right\rVert}

\def \M{M_0}
\def \sig{\sigma}
\def \Mp{M}
\def \Sig{\Sigma}

\def \y{y}
\def \My{M_\y}
\def \epsilony{\epsilon_\y}

\def \U{U}
\def \Uc{U_c}
\def \V{V}
\def \Vc{V_c}

\def \Lam{\Lambda}
\def \lam{\lambda}

\def \Mhat{\hat{M}}
\def \s{s}
\def \shat{\hat{s}}

\def \lamhat{\hat{\lambda}} 
\def \Vhat{\hat{\V}} 
\def \Vhatc{\hat{\V}_{c}} 
\def \Uhat{\hat{\U}} 
\def \Uhatc{\hat{\U}_{c}} 

\def \p{p}
\def \phat{\hat{\p}}

\def \tauhatp{\hat{\tau}_{\p}}
\def \tauhatphat{\hat{\tau}_{\phat}}

\def \bu{\mathbf{u}}
\def \bv{\mathbf{v}}
\newcommand{\blam}{\boldsymbol{\mathbf{\lambda}}}

\def \n{n}
\def \d{d}
\def \rank{r}
\def \m{m}
\def \L{L}
\def \O{\mathcal{O}}
\def \Q{\mathcal{Q}}
\def \I{\mathbbm{1}} 
\def \P{\mathcal{P}_{\Omega}}

\begin{document}


\renewcommand{\baselinestretch}{1.2}

\markright{ \hbox{\footnotesize\rm Statistica Sinica
}\hfill\\[-13pt]
\hbox{\footnotesize\rm
}\hfill }

\markboth{\hfill{\footnotesize\rm Juhee Cho, Donggyu Kim, and Karl Rohe} \hfill}
{\hfill {\footnotesize\rm Asymptotic theory for LRMC} \hfill}

\renewcommand{\thefootnote}{\arabic{footnote}}
$\ $\par


\fontsize{10.95}{14pt plus.8pt minus .6pt}\selectfont
\vspace{0.8pc}
\centerline{\large\bf Asymptotic Theory for Estimating the Singular Vectors and }
\vspace{2pt}
\centerline{\large\bf Values of a Partially-observed Low Rank Matrix with Noise}
\vspace{.4cm}
\centerline{Juhee Cho, Donggyu Kim, and Karl Rohe\footnote{This research is supported by NSF grant DMS-1309998 and ARO grant W911NF-15-1-0423.
}}
\vspace{.4cm}
\centerline{\it University of Wisconsin-Madison}
\vspace{.55cm}
\fontsize{9}{11.5pt plus.8pt minus .6pt}\selectfont


\begin{quotation}
\noindent {\it Abstract:}
Matrix completion algorithms recover a low rank matrix from a small fraction of the entries, each entry contaminated   with additive errors.
In practice, the singular vectors and singular values of the low rank matrix play a pivotal role for statistical analyses and inferences.
This paper proposes estimators of these quantities and studies their asymptotic behavior.
Under the setting where the dimensions of the matrix increase to infinity and the probability of observing each entry is identical, Theorem \ref{thm1} gives the rate of convergence for the estimated singular vectors; Theorem \ref{corol1} gives a multivariate central limit theorem for the estimated singular values.
Even though the estimators use only a partially observed matrix, 
	they achieve the same rates of convergence as the fully observed case.
These estimators combine to form a consistent estimator of the full low rank matrix 
	that is computed with a non-iterative algorithm.
In the cases studied in this paper, this estimator achieves the minimax lower bound in \cite{koltchinskii2011}.
The numerical experiments corroborate our theoretical results.\par

\vspace{9pt}
\par
\end{quotation}\par

\def\thefigure{\arabic{figure}}
\def\thetable{\arabic{table}}

\fontsize{10.95}{14pt plus.8pt minus .6pt}\selectfont


\section{Introduction} \label{intro}
The matrix completion problem arises in several different machine learning and engineering applications, ranging from collaborative filtering (\cite{rennie2005}), to computer vision (\cite{weinberger2006}), to positioning (\cite{montanari2010}), and to recommender systems (\cite{bennett2007}).
The  literature has  established a sizable body of algorithmic research (\cite{rennie2005,keshavan2009,cai2010,mazumder2010,hastie2014,cho2015nips})
and theoretical results 
(\cite{fazel2002,srebro2004,candes2009recht,candes2010plan,keshavan2010a,
recht2011,gross2011,negahban2011,koltchinskii2011,rohde2011,koltchinskii2011solo,candes2011tight,
negahban2012,cai2013,davenport2014,chatterjee2014}).
This extant literature is primarily focused on estimating the unobserved entries of the matrix.
In several of these previous estimation techniques, the algorithms first estimate the singular vectors and singular values of the low rank matrix.  
Also, based upon classical multivariate statistics,  these singular vectors and singular values can serve various types of statistical analyses and inferences.
For example, the overarching aim in the Netflix problem was to predict the unobserved film ratings and the previous algorithms and theories served this purpose.  However, if one wishes to interpret the resulting model predictions, then the estimated singular vectors and singular values can provide insights on (i) the main latent factors of film preferences and (ii) their relative strengths, respectively.  In the Netflix example,
\begin{quote} ``The first factor has on one side lowbrow comedies and horror movies, aimed at a male or adolescent audience (Half Baked, Freddy vs. Jason), while the other side contains drama or comedy with serious undertones and strong female leads (Sophie's Choice, Moonstruck). The second factor has independent, critically acclaimed, quirky films (Punch-Drunk Love, I Heart Huckabees) on one side, and mainstream formulaic films (Armageddon, Runaway Bride) on the other side.'' (\cite{koren2009matrix})
\end{quote}
This inference is based upon the leading singular vectors of the estimated matrix. 
To the best of our knowledge, no previous  research has studied the statistical properties of the estimated singular vectors and singular values.  

This paper proposes estimators of the singular vectors and singular values of the low rank matrix 
	as well as an estimator of the low rank matrix itself.
First, Lemma \ref{lem0} studies the singular vectors and singular values of a partially observed matrix that simply  substitutes zeros for the unobserved entries; the resulting estimators are biased.  
The proposed estimators adjust for this bias.  
Theorem \ref{thm1} 	finds the convergence rate for the bias-adjusted singular vector estimators and 
	Theorem \ref{corol1} gives a multivariate central limit theorem for the bias-adjusted singular value estimators.  
Despite the fact that the proposed estimators are built upon a partially observed matrix,
	they converge at the same rate as the standard estimators built from a fully observed matrix 
	up to a constant factor which depends on the probability of observing each entry.
Combining the proposed singular vector and value estimators,
	Section \ref{consist} gives a one-step consistent estimator of the low rank matrix
	which does not iterate over several singular value decompositions or eigenvalue decompositions.  
The mean squared error of this estimator achieves the minimax lower bound in Theorems 5-7 (\cite{koltchinskii2011}).

The rest of this paper is organized as follows. 
Section \ref{setup} describes the model setup. 
Section \ref{estimate} shows that the singular vectors and singular values of a partially observed matrix are biased 
	and suggests a bias-adjusted alternative.
Section \ref{converge} finds (1) the convergence rates of the estimated singular vectors and (2)  the asymptotic distribution of the estimated singular values.
Section \ref{consist} proposes and studies a one-step consistent estimator of the full matrix.
Section \ref{simul} corroborates the theoretical findings with numerical experiments. 
Finally, Section \ref{proofs} provides the proofs of our main theoretical results.
The proofs of the other results are collected in the Appendix.


\section{Model setup} \label{setup}

The underlying matrix that we wish to estimate is an $\n\times\d$ matrix $M_0$ with rank $\rank$.  By singular value decomposition (SVD), 
\begin{equation} \label{mod}
\M = \U\Lam\V^T,
\end{equation}
for orthonormal matrices 
$\U = (\U_1,\ldots,\U_\rank) \in \real^{\n \times \rank} \ \mbox{ and }  \ 
	\V = (\V_1,\ldots,\V_\rank) \in \real^{\d \times \rank}$
containing the left and right singular vectors, and a diagonal matrix 
$\Lam = \diag (\lam_1, \ldots, \lam_\rank) \in \real^{\rank \times \rank}$
containing the singular values.
$\M$ is corrupted by noise $\epsilon \in \real^{\n \times \d}$,
	where the entries of $\epsilon$ are i.i.d. sub-Gaussian random variables with mean zero and variance $\sig^2$.  
Let $\y \in \{0,1\}^{\n \times \d}$ be such that $\y_{kh} = 1$ if the $(k,h)$-th entry of $\M+\epsilon$ is observed and $\y_{kh} = 0$ otherwise.  
The entries of $\y$ are i.i.d. Bernoulli($\p$) and independent of the entries of $\epsilon$.
Thus, the total number of observed entries in $\M+\epsilon$ is a Binomial($\n\d,\p$) random variable. 
We observe $\y$ and the partially observed matrix $\Mp  \in\real^{\n\times\d}$, where
\begin{equation*}
\Mp_{kh} 
= \big[\y \cdot ({\M} + \epsilon)\big]_{kh} 
		=  \left\{\begin{matrix}
				{\M}_{kh}+\epsilon_{kh} & \text{if observed ($\y_{kh}=1$)}\\ 
				0 &  \text{otherwise ($\y_{kh}=0$)} 
				\end{matrix}\right. 
\end{equation*}
for $1\le k\le\n$ and $1\le h\le\d$. 
Throughout the paper, it is presumed that $\rank \ll \d \leq \n$.  
Moreover, the entries of $\M$ are bounded in absolute value by a constant $\L>0$.

\begin{remark}
Depending on the case, the noise $\epsilon$ can be related to the measurement system so that 
	assuming that there exist errors for unobserved entries does not make sense.
Hence, assume a hierarchical model as follows;
\begin{eqnarray*}
&&\epsilon_{ij}| \y_{ij}=0 = 0 \;\text{ a.s.,}\cr
&&\epsilon_{ij}| \y_{ij}=1 \sim \text{ subgaussian, and}\cr
&&\quad\;\, \y_{ij} \quad\;\;\,\sim \text{ i.i.d. Bernoulli}(\p).
\end{eqnarray*}
In this setting, the results obtained in this paper would still hold 
	although it may require more techniques or minor changes in the proof.
For simplicity of the paper, we only focus on the original setting. 
\end{remark}


\section{Estimation of singular values and vectors of $\M$} \label{estimate}

The vast majority of previous estimators of $\M$ have been initialized with $\Mp$, in effect imputing the missing values with zero. 
In this section, we study the properties of singular vectors and values of $\Mp$. 
This suggests alternative estimators of the singular vectors and values of $\M$.


\subsection{Properties of singular values and vectors of $\Mp$} \label{propertyM}

Define
$$\hat{\Sig}:=\Mp^T \Mp \quad\text{and}\quad \hat{\Sig}_t:=\Mp \Mp^T.$$
Then, the eigenvectors of $\hat{\Sig}$ and $\hat{\Sig}_t$ are the same as the right and left singular vectors of $\Mp$, respectively, and
	the squared root of eigenvalues of $\hat{\Sig}$ are the same as the singular values of $\Mp$.
The following lemma shows that $\hat{\Sig}$ and $\hat{\Sig}_t$ are biased estimators of $\M^T\M$ and $\M\M^T$, respectively.

\begin{lemma} \label{lem0}
Under the model setup in Section \ref{setup}, we have
\begin{equation} \label{expect-Sigp}
\expect \, \hat{\Sig} = \p^2 \M^T\M + \p(1-\p)\,\diag(\M^T\M) + \n\p\sig^2 I_\d,
\end{equation}
and similarly, 
\begin{equation} \label{expect-Sigpt}
\expect \, \hat{\Sig}_t = \p^2 \M\M^T + \p(1-\p)\,\diag(\M\M^T) + \d\p\sig^2 I_\n,
\end{equation}
where $I_\d$ and $I_\n$ are $\d\times\d$ and $\n\times\n$ identity matrices, respectively.
\end{lemma}


The proof of this lemma is in Appendix \ref{apdx0}.
The right-hand side of \eqref{expect-Sigp} contains terms  beyond $\p^2 \M^T \M$
	and they make the singular vectors and singular values of $\Mp$ biased estimators of the singular vectors and values of $\M$.
While the bias coming from $\n\p\sig^2 I_\d$ 
is manageable\footnote{This term does not change the singular vectors of $\expect \,\hat{\Sig}$; it merely increases each singular value by  $\n\p\sig^2$.},
	the bias coming from $\p(1-\p)\,\diag(\M^T\M)$ is not.  
The same applies to $\hat{\Sig}_t $ in \eqref{expect-Sigpt}.

To get rid of the terms producing unmanageable biases, we define $\hat{\Sig}_\p$ and $\hat{\Sig}_{\p t}$ and their eigenvectors and eigenvalues as follows,
\begin{equation}\label{Sigpp}
\begin{aligned}
\hat{\Sig}_\p :=& \hat{\Sig} - (1-\p)\,\diag(\hat{\Sig}) 
\cr
	=& (\V_\p,\V_{\p c}) \,\diag({\lam_\p^2}_1,\ldots,{\lam_\p^2}_\d)(\V_\p,\V_{\p c})^T, \;\text{ and}
\cr 
\hat{\Sig}_{\p t} :=& \hat{\Sig}_t - (1-\p)\,\diag(\hat{\Sig}_t)
\cr
	=& (\U_\p,\U_{\p c}) \,\diag({\lam_{\p t}^2}_1,\ldots,{\lam_{\p t}^2}_\n)(\U_\p,\U_{\p c})^T,
\end{aligned}
\end{equation}
where
\begin{eqnarray*}
\V_\p=({\V_\p}_1,\ldots,{\V_\p}_\rank) \in \real^{\d \times \rank}, && 
\V_{\p c}=({\V_\p}_{\rank+1},\ldots, {\V_\p}_{\d}) \in \real^{\d \times (\d-\rank)},\\
\U_\p=({\U_\p}_1,\ldots,{\U_\p}_\rank) \in \real^{\n \times \rank}, &&
\U_{\p c}=({\U_\p}_{\rank+1},\ldots,{\U_\p}_\n) \in \real^{\n \times (\n-\rank)}.\
\end{eqnarray*} 
The following proposition shows that $\hat{\Sig}_{\p}$ and $\hat{\Sig}_{\p t}$ adjust the bias.


\begin{proposition} \label{prop0}
Under the model setup in Section \ref{setup}, we have by eigendecomposition,
\begin{eqnarray*}
	&&\expect \, \hat{\Sig}_\p = \p^2 \M^T\M + \n\p^2\sig^2 I_\d =  ({\V}, {\Vc} ) \ddot{\Lam}_\p^2 ({\V}, {\Vc} )^T \text{ and} \cr
	&&  \expect \, \hat{\Sig}_{\p t} = \p^2 \M\M^T + \d\p^2\sig^2 I_\n  = ({\U}, {\Uc} ) \ddot{\Lam}_{\p t}^2 ({\U}, {\Uc} )^T,
\end{eqnarray*}
where 
	$\V$ and $\U$ are as defined in \eqref{mod}, 
	${\Vc} \in \real^{\d \times (\d-\rank)}$, ${\Uc} \in \real^{\n \times (\n-\rank)}$,
\begin{align*}
\ddot{\Lam}_\p^2 &= \diag ({\ddot{\lam_\p}^2}_1, \ldots, {\ddot{\lam_\p}^2}_\d) 
\\
	&= \diag (\p^2[{\lam_1^2}+\n\sig^2], \ldots, \,\p^2[{\lam_\rank^2} +\n\sig^2],\,\p^2\n\sig^2,\ldots, \, \p^2\n\sig^2)\in\real^{\d\times\d}, \text{ and}
\\
\ddot{\Lambda}_{\p t}^2 
	&= \diag (\p^2[{\lam_1^2}+\d\sig^2], \ldots, \p^2[{\lam_\rank^2}+\d\sig^2],\,\p^2\d\sig^2,\ldots, \,\p^2\d\sig^2) \in\real^{\n\times\n}.
\end{align*}
\end{proposition}

The proof of this proposition easily follows from Lemma \ref{lem0} and \eqref{Sigpp}.

Proposition \ref{prop0} shows that the top $\rank$ eigenvectors of $\expect \, \hat{\Sig}_\p$ and $\expect \, \hat{\Sig}_{\p t}$ are the same as the right and left singular vectors of $\M$, respectively. 
Also, the top $\rank$ eigenvalues of $\expect \, \hat{\Sig}_\p$ are easily adjusted to match the singular values of $\M$ as follows,
\begin{equation*}
\lam_i^2 = \frac{1}{\p^2}{\ddot{\lam_\p}^2}_i -\n\sig^2, \quad\text{for } i=1,\ldots,\rank.
\end{equation*}


\subsection{Estimators of singular values and vectors of $\M$}

The results in Proposition \ref{prop0} suggest plug-in estimators 
	using the leading eigenvectors and eigenvalues of $ \hat{\Sig}_\p$ and the leading eigenvectors of $\hat{\Sig}_{\p t}$ 
	as estimators of $\V$, $\Lam$, and $\U$, respectively.
However, since $\p$ is an unknown parameter in practice, 
	the proposed estimators use instead of $\p$ the proportion of observed entries in $\Mp$, $\phat$, which is defined as 
\begin{equation} \label{phat}
\phat = \frac{\sum_{k=1}^\n\sum_{h=1}^\d \y_{kh}}{\n\d}.
\end{equation}
Using $\phat$,  define $\hat{\Sig}_{\phat}$ and $\hat{\Sig}_{\phat t}$ as
\begin{equation} \label{Sigpphat}
\hat{\Sig}_{\phat} := \hat{\Sig} - (1-\phat)\,\diag(\hat{\Sig}) \quad\text{and}\quad \hat{\Sig}_{\phat t} := \hat{\Sig}_t - (1-\phat)\,\diag(\hat{\Sig}_t).
\end{equation}
By eigendecomposition, 
\begin{equation}\label{SigpphatDecompose}
\hat{\Sig}_{\phat} = ({\Vhat}, {\Vhatc} ) \,\Lam_{\phat}^2\, ({\Vhat}, {\Vhatc} )^T 
	\quad\text{and}\quad 
	\hat{\Sig}_{\phat t} = ({\Uhat}, {\Uhatc} ) \,\Lam_{\phat t}^2 \, ({\Uhat}, {\Uhatc} )^T,
\end{equation}
where ${\Vhat} \in \real^{\d \times \rank}$, ${\Vhatc} \in \real^{\d \times (\d-\rank)}$, 
	$\Lam_{\phat}^2 = \diag (\lam_{\phat 1}^2, \ldots, {\lam_{\phat \d}^2}) \in \real^{\d \times \d}$, 
	${\Uhat} \in \real^{\n \times \rank}$, ${\Uhatc} \in \real^{\n \times (\n-\rank)}$, and 
	$\Lam_{\phat t}^2 = \diag ({\lambda_{\phat t 1}^2}, \ldots, {\lambda_{\phat t \n}^2}) \in \real^{\n \times \n}$. 
Then,  estimate  the left and right singular vectors, $\U$ and $\V$, of $\M$ by $\Uhat$ and $\Vhat$, respectively.
Also, estimate the singular values, $\lam_i, i=1,\ldots,\rank$, of $\M$ by 
\begin{equation} \label{lambdahat}
\lamhat_i = \sqrt{\frac{1}{\phat^2}\big({\lam_{\phat}^2}_i - \tauhatphat\big)} \;\;\text{ for $i = 1, \ldots, \rank$},
\end{equation}
where $\tauhatphat = \frac{1}{\d-\rank} \trace\(\Vhatc^T \hat{\Sig}_{\phat} \Vhatc\). $

For any $A \in \real^{\n\times\d}$, let the $i$-th left singular vector of $A$ be denoted by $\bu_i(A)$,  
	the $i$-th right singular vector of $A$ by $\bv_i(A)$, and
	the top $i$-th singular value of $A$ by $\blam_i(A)$
	for $i =1, \ldots,\d$.
Then, Algorithm \ref{alg:UhatVhatLamhat} summarizes the steps to compute the proposed estimators of the singular values and vectors of $\M$.
\begin{algorithm}[ht]
\caption{\;Estimators of $\U_i$, $\V_i$, and $\lambda_i$  for $i=1,\ldots, \rank$}
\begin{algorithmic}
    \Require{$\Mp$, $\y$, and $\rank$}
    \State 
        $\phat \gets \frac{1}{\n\d}\sum_{k=1}^\n \sum_{h=1}^\d \y_{kh}$
    \State 
        ${\hat{\Sig}}_{\phat} \gets \Mp^T\Mp - (1-\phat) \diag(\Mp^T\Mp)$
    \State 
        ${\hat{\Sig}}_{t \phat} \gets \Mp\Mp^T - (1-\phat) \diag(\Mp\Mp^T)$
    \State 
        $\Vhat_i \gets \bv_i({\hat{\Sig}}_{\phat}), \quad \forall i \in \{1,\ldots,\rank\}$
    \State 
        $\Uhat_i \gets \bu_i({\hat{\Sig}}_{\phat t}), \quad \forall i \in \{1,\ldots,\rank\}$
    \State 
        $\tauhatphat \gets \frac{1}{\d-\rank} \sum_{i=\rank+1}^\d \blam_i({\hat{\Sig}}_{\phat})$
    \State 
        $\lamhat_i \gets \frac{1}{\phat}\sqrt{\blam_i({\hat{\Sig}}_{\phat}) - \tauhatphat}, \quad \forall i \in \{1,\ldots,\rank\}$
    \\
    \Return{$\Vhat_i$, $\Uhat_i$, and $\lamhat_i$ for $i=1,\ldots, \rank$}
\end{algorithmic}
\label{alg:UhatVhatLamhat}
\end{algorithm}


\section{Asymptotic theory} \label{asympt}

This section investigates the statistical properties of the estimators  proposed in \eqref{SigpphatDecompose} and \eqref{lambdahat}.


\subsection{Convergence rate of the estimated singular vectors and asymptotic distribution of the estimated singular values} \label{converge}

Let $x=(x_1, \ldots, x_n)^T$ be a $n$-dimensional vector and $ A=(A_{kh})$ a $n\times d$ matrix. 
Then, the $\ell_p$-norm is defined as follows,
\begin{equation*}
\left \|  x \right \| _p =\( \sum_{i=1} ^{p} |x_i|^{p} \)^{1/p}, \quad\text{and}\quad \left \|  A \right \|_p =\sup \{ \left \|  A x \right \| _p, \left \|  x \right \| _p =1\}, \;\; p=1,2, \infty.
\end{equation*}
The spectral norm $\|  A \| _2$ is a square root of the largest eigenvalue of $ A  A^T$,
\begin{equation*}
\left \| A \right \| _1 =\max_{1\leq h \leq d} \sum_{k=1} ^{n} |A_{kh}|, \;\text{ and }\; \left \| A \right \| _{\infty} =\max_{1\leq k \leq n} \sum_{h=1} ^{d} |A_{kh}|.
\end{equation*}
The squared Frobenius norm is defined by $\left \| A \right \| _F ^2 = \trace\(A^T A\)$, the trace of $A^T A$.
We denote by $c>0$ and $C>0$ generic constants that are free of $\n$, $\d$, and $\p$, and different from appearance to appearance. 

To measure how close the proposed estimator $\hat\V$ is to $\V$ (or, $\hat\U$ to $\U$), 
	we introduce a classical notion of distance between subspaces.
Let $\ran(Z_1)$ denote a column space spanned by $Z_1 \in \real^{\d\times\rank}$ and $\ran(Z_2)$ by $Z_2 \in \real^{\d\times\rank}$.
Then, to measure the dissimilarity between $\ran(Z_1)$ and $\ran(Z_2)$,  consider the following loss function 
$$\| \sin (Z_1, Z_2 ) \| _F ^2 = \| \sin \Theta (\ran(Z_1), \ran(Z_2) ) \| _F ^2,$$
where $\sin \Theta (\ran(Z_1), \ran(Z_2) )$ is a diagonal matrix of singular values (canonical angles) 
	of $P_1 P_2 ^{\perp}$ with orthogonal projections $P_1$ and $P_2$ of $Z_1$ and $Z_2$, respectively.
Here $P^{\perp}=I - P$.
The canonical angles generalize the notion of angles between lines and are often used to define the distance between subspaces.
If the columns of $Z_1$ and $Z_2$ are singular vectors, 
	$\ran(Z_1)$ and $\ran(Z_2)$ have projections $P_1=Z_1 Z_1^T$ and $P_2=Z_2Z_2 ^T$, 
	respectively, and $\|\sin (Z_1, \hat{Z_2} )\|_F ^2 = \|Z _1Z _1 ^T (Z _2Z _2 ^T) ^{\perp}\|_F ^2 
	= \frac{1}{2}\|Z _1Z _1 ^T -Z _2Z _2 ^T\|_F ^2$.
Proposition 2.2 in \cite{vu2013} relates this subspace distance 
	to the Frobenius distance 
\begin{equation} \label{sineDist}
	\frac{1}{2} \inf_{\O \in \mathbb{V}_{\rank,\rank}} \smallnorm{Z_1 - Z_2\O}_F^2 
		\le \smallnorm{ \sin (Z_1 ,Z_2 )}_F^2 
			\le \inf_{\O \in \mathbb{V}_{\rank,\rank}} \smallnorm{Z_1 - Z_2\O}_F^2 ,
\end{equation}
where $\mathbb{V}_{\rank,\rank} = \{O \in \real^{\rank\times\rank} : O^TO = I_{\rank}$ and $OO^T = I_{\rank}\}$ denotes the Stiefel manifold of $\rank \times \rank$ orthonormal matrices.
In other words, the distance between two subspaces corresponds to the minimal distance between their orthonormal bases.

\begin{assumption}\label{assume1}
~
\begin{enumerate}
	\item [(1)] $\lam_{i} = b_i \sqrt {\n\d}, i = 1, \ldots, \rank$, where $\frac{1}{c}\le b_i \le c$ for a constant $c>0$;
	\item [(2)] there exists a constant $\m \in \{1, \ldots,\rank \}$ such that ${b_{m}} > {b_{m+1}}$, where ${b}_{\rank+1}=0$;
	\item [(3)] $\d \le \n \le e^{\d^{\alpha}}$ for a constant $\alpha<1$ free of $\n$, $\d$, and $\p$.
\end{enumerate}
\end{assumption}
\begin{remark} \label{rmk:justfyND}
To motivate Assumption \ref{assume1} (1),
suppose that a non-vanishing proportion of entries of $\M$ contains non-vanishing signals (i.e. ${\M}_{kh}^2 \ge c_0$ for some constant $c_0>0$)
	and that the rank of $\M$ is fixed. 
Then, $$\sum_{k=1}^\n \sum_{h=1}^\d {\M}_{kh}^2=\norm{\M}_F^2\ge c\n\d$$ for some constant $c>0$. 
Because the squared Frobenius norm is also the sum of the squared singular values of $\M$, 
	the order of the singular values of $\M$ should be $\sqrt{\n\d}$ (see also \cite{fan2013}).  
Assumption \ref{assume1}(1) may seem uncommon in the matrix completion literature,
	but consider the widely-used assumption (II.2) in \cite{candes2010plan},	
$$  \max_{1\le k\le \n} |\U_{ik}| \le \sqrt{C/\n} \;\;\text{and}\;\;   \max_{1\le h\le \d} |\V_{ih}| \le \sqrt{C/\d}$$
for $i=1,\ldots,\rank$ and a constant $C\ge 1$, which prevents spiky singular vectors.
Under the model setup in Section \ref{setup} where the entries of $\M$ are bounded in absolute value by a constant $\L>0$, 
	this implies Assumption \ref{assume1}(1). 
\end{remark}
The following theorem shows the convergence of $\hat\V$ to $\V$ and $\hat\U$ to $\U$.


\begin{theorem} \label{thm1}
Under the model setup in Section \ref{setup} and Assumption \ref{assume1}, 
let $\Vhat^{(\m)}$ and $\Uhat^{(\m)}$ be the first $m$ columns of  $\Vhat $ and $\Uhat$ defined in \eqref{SigpphatDecompose}, respectively, 
and let $\V^{(\m)}$ and $\U^{(\m)}$ be the first $m$ columns of $\V $ and $\U$ defined in \eqref{mod}, respectively.
Then, for large $\n$ and $\d$,
\begin{eqnarray} \label{Vconsist} 
\expect\norm{\sin \big( \Vhat^{(\m)}, \V^{(\m)} \big)}_F^2 
	\leq  \frac{C_1\, n^{-1}}{\p\,({b}_{\m}^2 - {b}_{\m+1}^2)^2}
\end{eqnarray}
and
\begin{eqnarray} \label{Uconsist}
\expect\norm{\sin \big( \Uhat^{(\m)}, \U^{(\m)} \big)}_F^2 
	\leq  \frac{C_2 \, d^{-1}}{\p\,({b}_{\m}^2 - {b}_{\m+1}^2)^2},
\end{eqnarray}
where $C_1$ and $C_2$ are generic constants free of $n, d,$ and $p$.
\end{theorem}


The proof of this theorem is in Section \ref{proofs:thm1}.


\begin{remark} \label{rmk2}
As long as $\frac{\p\,\d}{\log\n} \to \infty$, the convergence rates in Theorem \ref{thm1} will hold.
Hence, under the setting where $\p$ goes to zero, if $\d/\log\n$ diverges fast enough that $\frac{\p\,\d}{\log\n} \to \infty$, 
	we can still obtain the results in Theorem \ref{thm1}.
\end{remark}
\begin{remark}
Despite the fact that $\Vhat^{(\m)}$ is built on a partially observed matrix $\Mp$, 
	Theorem \ref{thm1} gives the convergence rate $\frac{n^{-1/2}}{({b}_{\m}^2 - {b}_{\m+1}^2) }$
	which is the standard convergence rate for eigenvectors (\cite{anderson1958}). 
The effect of the partial observations appears in the denominator of the right-hand side of \eqref{Vconsist} as $\p$.
A similar discussion applies to $\Uhat^{(\m)}$ in \eqref{Uconsist}.
\end{remark}

The next theorem shows the asymptotic distribution of $\lamhat_i^2$ centered around $\lam_i^2$.
\begin{theorem} \label{thm2}
Suppose $\n\d^{-1}\to\infty$. Then, under the model setup in Section \ref{setup} and Assumption \ref{assume1}, we have
\begin{equation*}
\frac{\sum_{i=1}^{\m} {\lamhat}_i^2  - \sum_{i=1}^{\m} \lam^2_i }{ \sqrt{nd} \sigma_\lambda}  
		\,  \rightarrow  \, \mathcal{N}(0,1) \, \text{ in distribution, }\quad\text{as $\n$ and $\d$} \to \infty.
\end{equation*}
where
\begin{eqnarray*}
\sigma_\lambda^2
= 
\frac{4(1-\p)}{\p} \Bigg\{ \sum_{k=1}^\n \sum_{h=1}^\d {\M}_{kh}^2 \bigg( \sum_{i=1}^\m b_i\U_{ik}\V_{ih} \bigg)^2
	-\bigg(\sum_{i=1}^\m  b^2_i\bigg)^2 \Bigg\} + \frac{4\sig^2}{\p} \sum_{i=1}^\m b_i^2,
\end{eqnarray*}
$\U_{ik}$ is the $k$-th entry of $\U_i$, and $\V_{ih}$ is the $h$-th entry of $\V_i$.
\end{theorem}


The proof of this theorem is in Section \ref{proofs:thm2}.
\begin{remark} \label{rmk:asympDist}
As long as $\frac{\p\,\d}{\log\n} \to \infty$ and $\p \n \d^{-1} \to \infty$, the asymptotic normality result in Theorem \ref{thm2} will hold.
Hence, under the setting where $\p$ goes to zero, if $\d/\log\n$ and $\n/\d$ diverge fast enough that $\frac{\p\,\d}{\log\n} \to \infty$ and $\p\n/\d \to \infty$, 
	we can still obtain the results in Theorem \ref{thm2}.
\end{remark}
\begin{remark}
Theorem \ref{thm2} shows that the convergence rate of $\sum_{i=1}^{\m} {\lamhat}_i^2$  is $\sqrt{\n\d}$. 
Considering Assumption \ref{assume1}(1), it is an optimal rate.
However, since the results are based on partially observed entries, the asymptotic variance, $\sigma_\lambda^2$, increases with the rate $p^{-1}$.
For example, when we have a fully-observed matrix, 
	$\sigma_\lambda^2$ simply becomes $4\sig^2 \sum_{i=1}^\m b^2_i$
	which is a lower bound for $\sigma_\lambda^2$.
\end{remark}


One of the main purposes of this paper is to investigate asymptotic behaviors of the estimators of the singular values of $\M$.
An application of the proof of Theorem \ref{thm2} and the delta method provides a multivariate central limit theorem for $\hat \lambda_1, \dots, \hat \lambda_r$. 
\begin{theorem} \label{corol1}
Suppose that $${b_{i}} > {b_{i+1}} \;\text{ for all } i \in \{1, \ldots, \rank\} \quad\text{ and }\quad \n\d^{-1}\to\infty.$$ 
Then, under the model setup in Section \ref{setup} and Assumption \ref{assume1}, we have
\begin{equation*}
\Upsilon^{-1/2}
\begin{pmatrix}
\hat\lam_1-\lam_1\\ 
\vdots \\ 
\hat\lam_r-\lam_r
\end{pmatrix}
\,  \rightarrow  \, \mathcal{N}\(0, I_{\rank} \) \, \text{ in distribution, }\quad\text{as $\n$ and $\d$} \to \infty,
\end{equation*}
where $\Upsilon = \Upsilon^T \in \real^{\rank\times\rank}$ consists of 
\begin{eqnarray}\label{asympVar}
\Upsilon_{ij}
= 
\begin{cases}
\frac{(1-\p)}{\p} \( \sum_{k=1}^\n \sum_{h=1}^\d {\M}_{kh}^2 \U_{ik}^2\V_{ih}^2
	- b_i^2 \) + \frac{\sig^2}{\p} & \text{ if } i=j \\ 
\frac{(1-\p)}{\p}\(\sum_{k=1}^\n \sum_{h=1}^\d {\M}_{kh}^2 \U_{ik}\V_{ih}\U_{jk}\V_{jh} -b_ib_j \) & \text{ if } i\ne j .
\end{cases}
\end{eqnarray}
Thus, $| {\lamhat}_i - \lam_i |= O_p\(\frac{1}{\sqrt{\p}}\)$.
\end{theorem}


\begin{remark}\label{rmk:thm3}
As in case of Theorem \ref{thm2} (see Remark \ref{rmk:asympDist}), 
	as long as $\frac{\p\,\d}{\log\n} \to \infty$ and $\p \n \d^{-1} \to \infty$, the asymptotic normality result in Theorem \ref{corol1} will hold.
Note that Theorems \ref{thm2} and \ref{corol1} require an additional condition, $\p \n \d^{-1} \to \infty$, to the condition required for Theorem \ref{thm1}, $\frac{\p\,\d}{\log\n} \to \infty$.
Under the setting where $\p$ is a constant, this additional condition implies that $\d/\n$ has to go to zero.
The rationale behind this is as follows. 
In Theorems \ref{thm2} and \ref{corol1}, 
	we find the limiting distribution on the singular values of $\M$ from a $\d\times \d$ matrix $\hat{\Sig}_{\phat}$, 
	while the total number of observations is $\n\d$.
That is, the size of our parameter space is $\d^2$ and 
	the total amount of information we can use to find asymptotic properties on the parameters is $\n\d$.
Since our observations are even noisy, we need an enough number of observations to achieve our goal.   
When $\d/\n \to 0$, 
	we can make the approximation errors in the singular values of $\hat{\Sig}_{\phat}$ negligible and 
	find the limiting distribution on the singular values of $\M$.
\end{remark}
\begin{remark}
The results of Theorems \ref{thm2} and \ref{corol1} help us to make statistical inference on the singular values of $\M$. 
For example, they open up possibilities for us to evaluate 
	how many factors are significant or how influential each factor is,
	by providing the distribution of the singular values.
\end{remark}


Theorems \ref{thm1}-\ref{corol1} show that the proposed estimators for $\U, \V, $ and $\lam_i$'s are asymptotically unbiased and have optimal convergence rates. 
With these well-developed estimators for the singular values and vectors of $\M$, 
the following section proposes a consistent estimator of $\M$. 


\subsection{A consistent estimator of $\M$} \label{consist}

Suppose that ${b_{i}} > {b_{i+1}}$ for all $i =1, \ldots,\rank$. 
Theorem \ref{thm1} and \eqref{sineDist} imply that $\Vhat_i$ and $\Uhat_i$ can estimate $\V_i$ and $\U_i$ 
	up to constant factors $\text{sign}(\langle \Vhat_i, \V_i \rangle)$ and $\text{sign}(\langle \Uhat_i, \U_i \rangle)$, respectively.
Let $\s_0 = (\s_{01}, \ldots, \s_{0r} ) \in \{-1, 1\}^\rank$ be 
\begin{equation} \label{s0}
\s_{0i} = \text{sign}(\langle \Vhat_i, \V_i \rangle)\,  \text{sign}(\langle \Uhat_i, \U_i \rangle) 
\quad\text{for}\quad i \in \{1, \ldots, \rank\}.
\end{equation} 
Then, $\Mhat (\s_0) = \sum_{i=1}^\rank \s_{0i} \, \lamhat_i \Uhat_i \Vhat_i^T$
becomes a consistent estimator of $\M$.
However, since $\s_0$ is an unknown parameter in practice, 
	we employ $\shat = (\shat_1, \ldots, \shat_\rank) \in \{-1, 1\}^\rank$ as an estimator of $\s_0$; 
\begin{equation} \label{shat}
\shat = \argmin_{\s \in \{-1, 1\}^\rank} \; \smallnorm{ \P\big(\Mhat(\s)\big) - \P\big(\Mp\big) }_F^2,
\end{equation} 
where $\Omega$ contains indices of the observed entries, $y_{kh}=1 \Leftrightarrow (k,h) \in \Omega$, and
$\P(A)$ for any $A\in\real^{\n\times\d}$ denotes the projection of $A$ onto $\Omega$,
$$\P(A)_{kh} = \left\{\begin{matrix}
A_{kh} & \text{if } (k,h) \in \Omega\\ 
0 & \text{if } (k,h) \notin \Omega
\end{matrix}\right.
\quad\text{ for } 1\le k\le \n \text{ and }1\le h\le \d.$$
Hence, the proposed estimator of $\M$ is 
\begin{equation} \label{Mhats}
\Mhat (\shat) = \sum_{i=1}^\rank \shat_{i} \, \lamhat_i \Uhat_i \Vhat_i^T.
\end{equation}


\begin{remark}
Finding $\shat$ as in \eqref{shat} requires $2^\rank$ computations. 
Hence, it can be a computational bottleneck or even impossible for a large $\rank$.
In such cases, we suggest an alternate way to find $\shat$ as follows;
$$\shat_{i}^{\;alternate} = \text{sign}(\langle \Vhat_i, \mathbf{v}_i(\Mp) \rangle) \; \text{sign}(\langle \Uhat_i, \mathbf{u}_i(\Mp) \rangle) \quad\text{for}\;\; i=1,\ldots,\rank.$$ 
Note that if we use $\V_i$ and $\U_i$ instead of $\mathbf{v}_i(\Mp)$ and $\mathbf{u}_i(\Mp)$, this gives us the true sign $\s_0$ in \eqref{s0}.
\end{remark}


In the following we show that $\Mhat(\shat)$ is a consistent estimator of $\M$ under certain conditions.
The steps to compute $\Mhat(\shat)$ using $\{ \Vhat_i,\Uhat_i,\lamhat_i\}_{i=1}^{\rank}$ from Algorithm \ref{alg:UhatVhatLamhat} 
	are summarized in Algorithm \ref{alg:Mhat}.
\begin{algorithm}[ht]
\caption{\;Estimator of $\M$}
\begin{algorithmic}
    \Require{$\Vhat_i$, $\Uhat_i$, and $\lamhat_i$ for $i=1,\ldots, \rank$}
    \State 
        $\shat \gets \argmin_{s \in \{-1,1\}^\rank} \norm{ \P\big( \sum_{i=1}^\rank s_i \lamhat_i \Uhat_i \Vhat_i^T\big) - \P\big(\Mp\big) }_F^2$
    \State $\Mhat(\shat) \gets \sum_{i=1}^\rank \shat_i \lamhat_i \Uhat_i \Vhat_i^{T}$
    \\
    \Return{$\Mhat(\shat)$}
\end{algorithmic}
\label{alg:Mhat}
\end{algorithm}

\begin{assumption}\label{assume2}
~
\begin{enumerate}
	\item [(1)] $\lim_{\n\to\infty,\d\to\infty} \prob\Big(\min_{\s \in \{-1, 1\}^\rank} \; \smallnorm{ \P\big(\Mhat(\s)\big) - \P\big(\Mp\big) }_F^2\\$
		$\hspace*{5.5cm} < \smallnorm{ \P\big(\Mhat(\s_0)\big) - \P\big(\Mp\big) }_F^2 \Big) = 0$;
	\item [(2)] ${b_{i}} > {b_{i+1}}$ for all $i =1, \ldots,\rank$.
\end{enumerate}
\end{assumption}
\begin{remark} 
When the rank $\rank$ is 1, it is more straightforward to understand Assumption \ref{assume2}(1). 
Assuming that $\s_0=1$, it means that 
\begin{eqnarray*}
&&\lim_{\n\to\infty,\d\to\infty} \prob\bigg(\smallnorm{ \P\big(-\lamhat \Uhat \Vhat^T\big) - \P\big(\Mp\big) }_F^2 
\cr
&&\hspace*{5.5cm}
	< \smallnorm{ \P\big(\lamhat \Uhat \Vhat^T\big) - \P\big(\Mp\big) }_F^2 \bigg) = 0.
\end{eqnarray*} 
That is, the probability that $\shat$ picks a different sign than the true sign $\s_0=1$ goes to zero with the dimensionality.
Given the asymptotic properties of our estimators $\lamhat, \Uhat$, and $\Vhat$, this is not an unreasonable assumption to make.
\end{remark} 


\begin{theorem} \label{thm3}
Under the model setup in Section \ref{setup} and Assumptions \ref{assume1}-\ref{assume2}, 
for any given $\eta >0$, there exists a constant $C_\eta >0$ such that for sufficiently  large $\n$, 
$$ \prob \( \frac{\p\,{b}_{\rank}^4}{\n} \norm{\Mhat(\shat)-\M}_F^2 \ge C_\eta \) \le \eta . $$
Or alternatively, 
$$\smallnorm{ \Mhat(\shat) - \M }_F^2 = \frac{1}{\p\,{b}_{\rank}^4}\,o_p\(h_\n \n\),  $$
where  $h_\n$ can be anything that diverges very slowly with the dimensionality, for example, $\log (\log \d )$.
\end{theorem}


The proof of this theorem is in Section \ref{proofs:thm3}.


\begin{remark}
As in case of Theorem \ref{thm1} (see Remark \ref{rmk2}), as long as $\frac{\p\,\d}{\log\n} \to \infty$, the convergence rates in Theorem \ref{thm3} will hold. 
If we let $\p=\frac{N}{\n\d}$ so that $N$ represents the number of observed entries in the population sense, 
	this condition implies that $\frac{N}{\n\log\n}\rightarrow\infty$. 
Therefore, for $\Mhat(\shat)$ to be consistent, the number of observed entries should increase at a faster rate than $\n\log\n$.
This is a comparable result to Theorem 1 in \cite{candes2010plan}.
\end{remark}
\begin{remark}
The additional condition, $\p \n \d^{-1} \to \infty$, required for Theorems \ref{thm2} and \ref{corol1} (see Remarks \ref{rmk:asympDist} and \ref{rmk:thm3}), 
	is not needed for Theorems \ref{thm1} and \ref{thm3}.
It means that if $\p$ is a constant, even though $\d/\n\rightarrow c$ for some $0<c\le1$ or $\d\le\n$, the results in  Theorems \ref{thm1} and \ref{thm3} will still hold, but the results in Theorems \ref{thm2} and \ref{corol1} will not.
\end{remark}
\begin{remark}
Theorem \ref{thm3} shows that 
$\frac{1}{nd} \smallnorm{\Mhat(\shat)-\M}_F^2$ is bounded by $Cp^{-1}d^{-1}$ for some constant $C>0$.
Under the setting where the rank of $\M$ is fixed as in this paper, 
	this is matched to the minimax lower bound in Theorems 5-7 (\cite{koltchinskii2011}). 
The previous estimators that obtain the minimax rate are computed 
	via semidefinite programs that require iterating over several SVDs.
However, the proposed estimator is a non-iterative algorithm. 
\end{remark}
\begin{remark}
\cite{chatterjee2014} established the minimax error rate for estimators of a general class of noisy incomplete matrices which extend beyond low rank matrix completion. 
In the regime studied herein, the convergence rate of our estimator of $\M$ is faster than the convergence rate in Theorem 2.1 (\cite{chatterjee2014}).
This is likely because we consider a smaller class of matrices, where the singular values of a low rank matrix have the divergence rate $\sqrt{nd}$ (Assumption \ref{assume1}(1)).
Remark \ref{rmk:justfyND} justifies this assumption in the setting of low rank matrix completion. 
\end{remark}


Throughout this paper, we have assumed that the rank, $\rank$, of $\M$ is known. 
However, it is an unknown parameter and needs to be estimated.
The following lemma proposes an estimator of $\rank$ and shows its consistency.


\begin{lemma} \label{thm5}
Let $C_\d>0$ such that $C_\d/\d\rightarrow 0$ and $C_\d\rightarrow \infty$, for example, $C_\d=c\log\d$ for any $c>0$. 
Also, let $\,\hat\rank = \big|\{ i\in \{1,\ldots,\d\}\; |\; \lam_{\phat i}^2 \,\ge\, \p^2\n\, C_\d \}\big|$ where $\lam_{\phat i}^2$ is defined in \eqref{SigpphatDecompose}. 
Then, for any given $\delta>0$, we have
$$\prob(\hat\rank = \rank)=1-O(\n^{-\delta}).$$
\end{lemma}

The proof of this lemma is in Appendix \ref{apdx4}.

\begin{remark}
Empirically to find $C_\d$ and $\hat\rank$ in Lemma \ref{thm5}, 
	we suggest using a scree plot of the singular values of $\hat{\Sig}_{\phat}$ in \eqref{Sigpphat}.
\end{remark}
\begin{remark}
As long as $C_\d$ satisfies $\sig^2 \p^2\n < \p^2\n\, C_\d \le (\sig^2+b_\rank^2\d)\, \p^2\n$, 
	consistency of $\hat\rank$ in Lemma \ref{thm5} will  hold.
However, in the finite sample case, if the noise level $\sig^2$ is larger than $b_\rank^2\d$,
	it can be difficult to observe a singular-value gap and determine $\hat\rank$ using the scree plot of the singular values of $\hat{\Sig}_{\phat}$.
\end{remark}



\section{Numerical experiments} \label{simul}


\subsection{Simulations}

This section studies the performance of the proposed estimators using
 several values of the dimension $\n$ and the probability $\p$.

To simulate $\M$, generate $A \in [-5,5]^{\n \times 2}$, $B \in [-5,5]^{\d \times 2}$ to contain i.i.d. Uniform$[-5,5]$ random variables and define 
$$M_0=A B^{T} \in \real^{\n\times\d}.$$
Each entry of $\M$ is observed with probability $\p$ and unobserved with probability $1-\p$.
The observed entries of $\M$ are corrupted by noise $\epsilon$ as defined in Section \ref{setup}, 
where $\epsilon_{kh}$ are i.i.d. $\mathcal{N}(0,1)$.  
The dimension $\n$ varies from 100 to 1000 and $\p$ from 0.1 to 1, while $\d=2\sqrt{\n}$.
Each simulation was repeated 500 times and the errors were averaged.

\begin{figure}[p]
\centering
\includegraphics[width=4.5in]{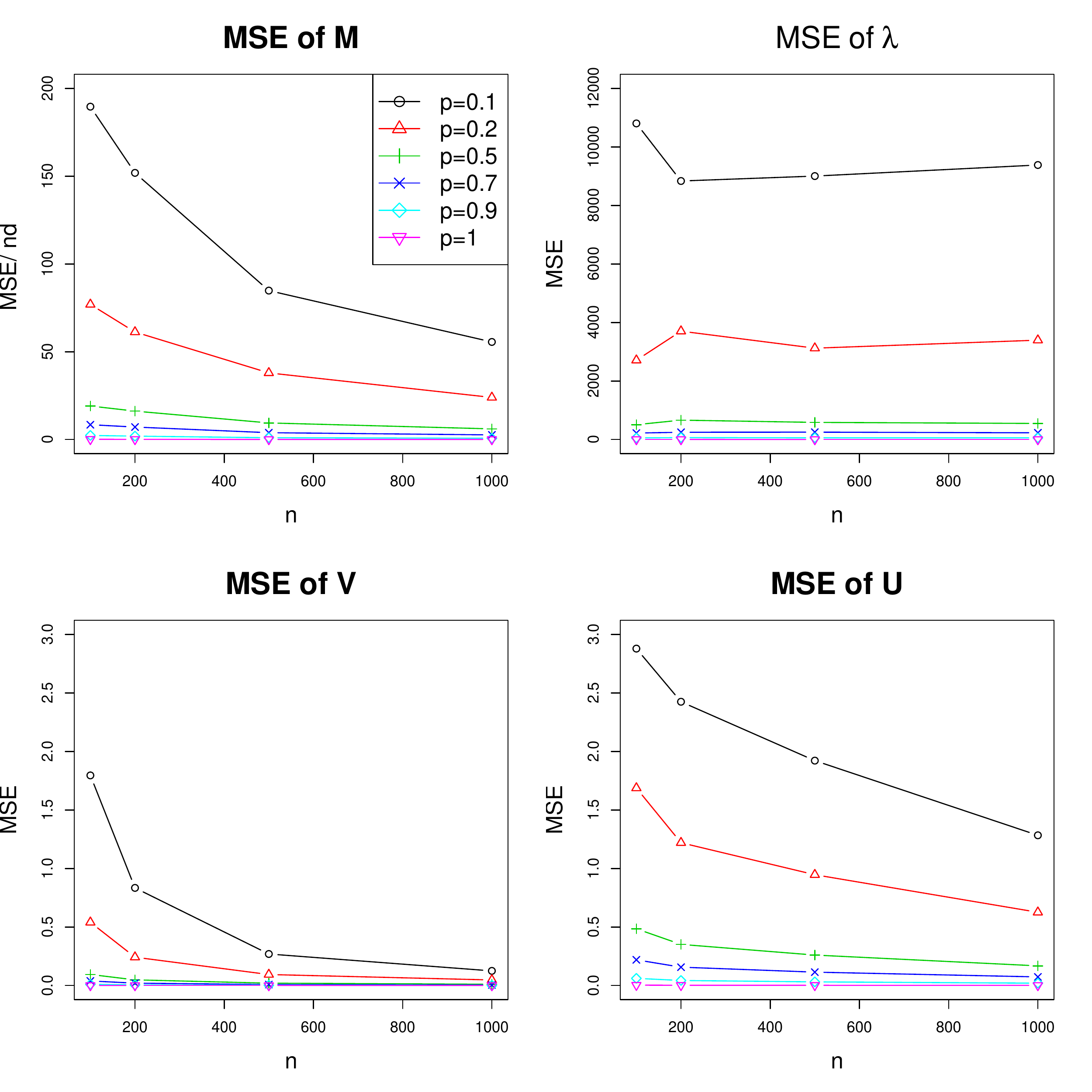}
\caption{The mean squared errors for six different values of $\p$ when $\n$ increases. Each point on the plots correspond to an average over 500 replicates.}
\label{fig:figure3}
\end{figure}

\begin{figure}[p]
\centering
\includegraphics[width=4.5in]{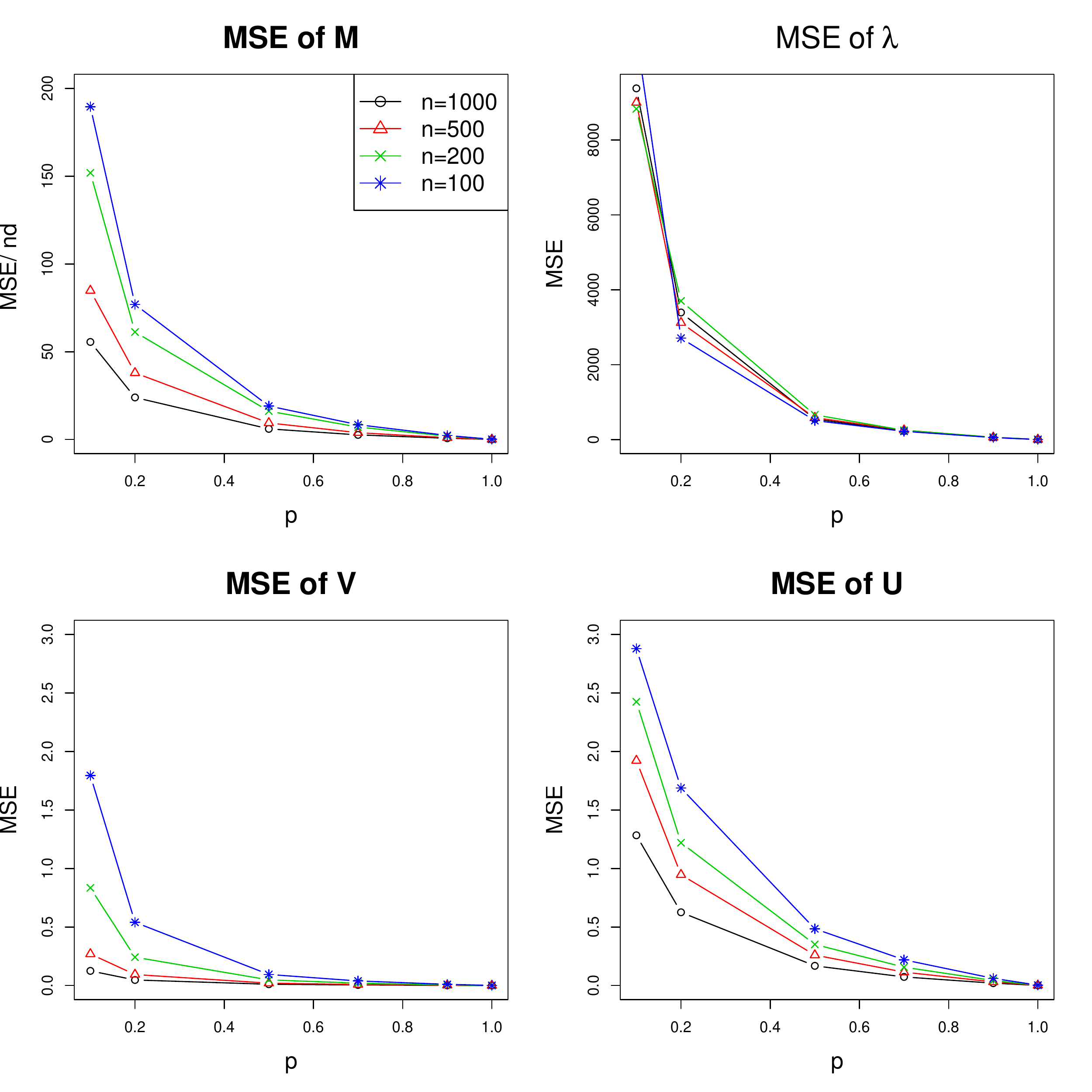}
\caption{The same mean squared errors as the ones in Figure 1 plotted for four different values of $\n$ when $\p$ increases. 
Each point on the plots correspond to an average over 500 replicates.}
\label{fig:figure2}
\end{figure}

\begin{figure}[p]
\centering
\includegraphics[width=5in]{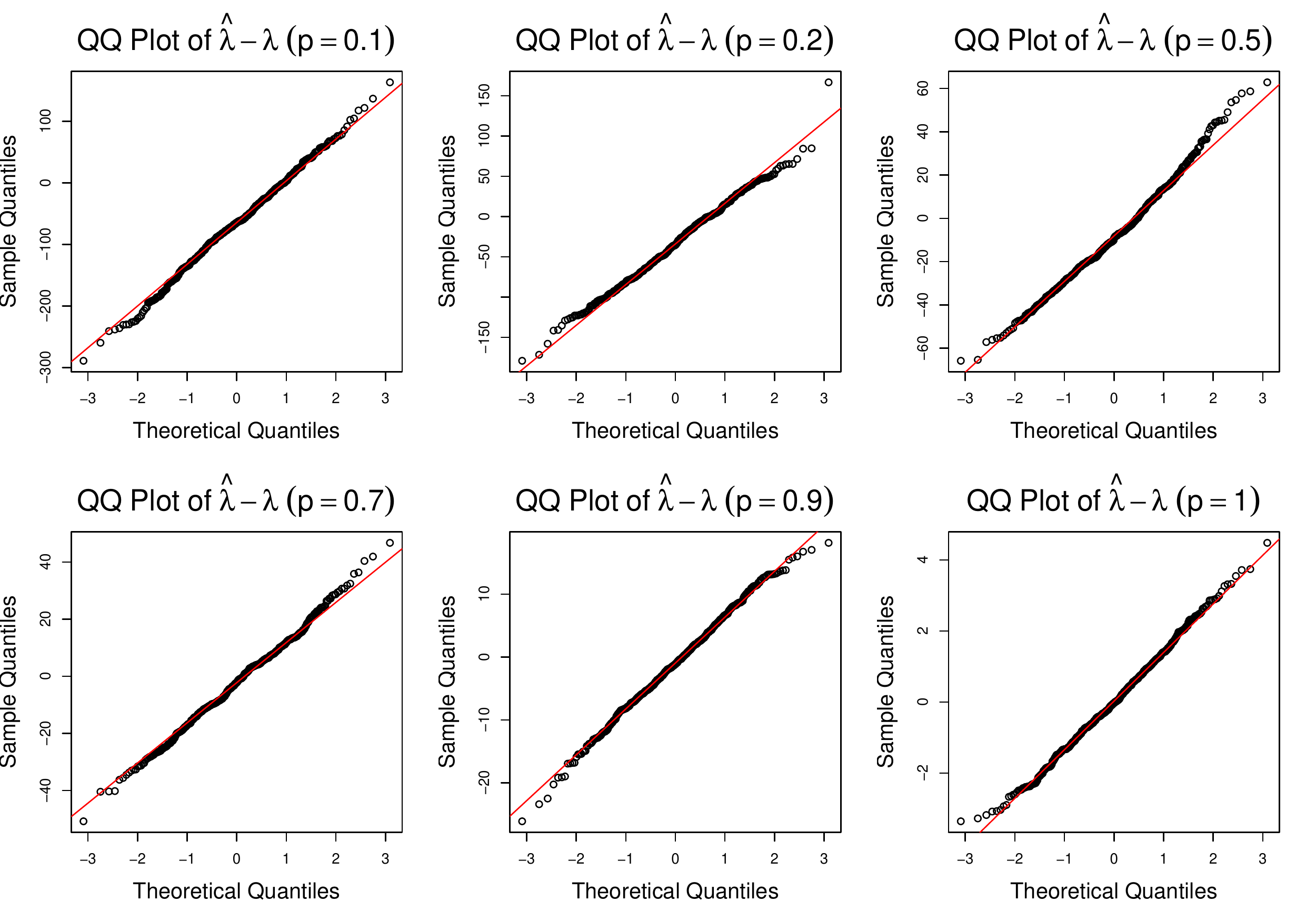}
\caption{Asymptotic normality of $\sum_{i=1}^2\lamhat_i - \sum_{i=1}^2\lam_i$ as $\p$ varies from 0.1 to 1. Across the plots, we fixed $\n$ to be 1000.}
\label{fig:figure1}
\end{figure}

Figures \ref{fig:figure3} and \ref{fig:figure2} summarize the resulting mean squared errors calculated by $\frac{1}{\n\d}\smallnorm{\Mhat(\shat)-\M}_F^2$, 
	$\smallnorm{\diag(\lamhat_1, \lamhat_2)-\Lam}_F^2$, $\smallnorm{\Vhat-\V}_F^2$, and $\smallnorm{\Uhat-\U}_F^2$, 
	when $\n$ and $\p$ increase along the $x$-axis, respectively.  
The MSE for $\Vhat$ decreases more rapidly than the MSE for $\Uhat$ and both MSEs decrease when $\p$ increases; 
	this is consistent with the results in Theorem \ref{thm1}. 
The MSE of $\Mhat$ decreases with the increase of $\n$ and $\p$. 
The MSE of $\lamhat$ stays stable over the changes of $\n$ 
	since it is measured on $\lamhat_i$ instead of $\lamhat_i^2$ (see Theorem \ref{corol1}), 
	but decreases with the increase of $\p$.

We further studied the asymptotic normality of $\sum_{i=1}^2\lamhat_i$ in Theorem \ref{corol1}. 
Figure \ref{fig:figure1} graphs the QQ plot of $\sum_{i=1}^2\lamhat_i - \sum_{i=1}^2\lam_i$,
	where the dimension $\n$ is fixed at 1000 and $\p$ varies from 0.1 to 1. 
This shows that the asymptotic normality holds across various values of $\p$. 


\subsection{A data example}

To illustrate the proposed estimation methods, this section analyzes the MovieLens 100k data (\cite{movielens100k}).
The data set consists of 100,000 ratings from 943 users and 1682 movies and each user has rated at least 20 movies. 
Taking this partially observed data matrix as $\Mp$, we computed $\hat{\Sigma}_{\phat}$ as in \eqref{Sigpphat}
	and plotted the scree plot of the singular values of $\hat{\Sigma}_{\phat}$ to determine $\hat\rank$. 
Figure \ref{fig:figure4} shows the result. 
\begin{figure}[!ht]
\centering
\includegraphics[width=4in]{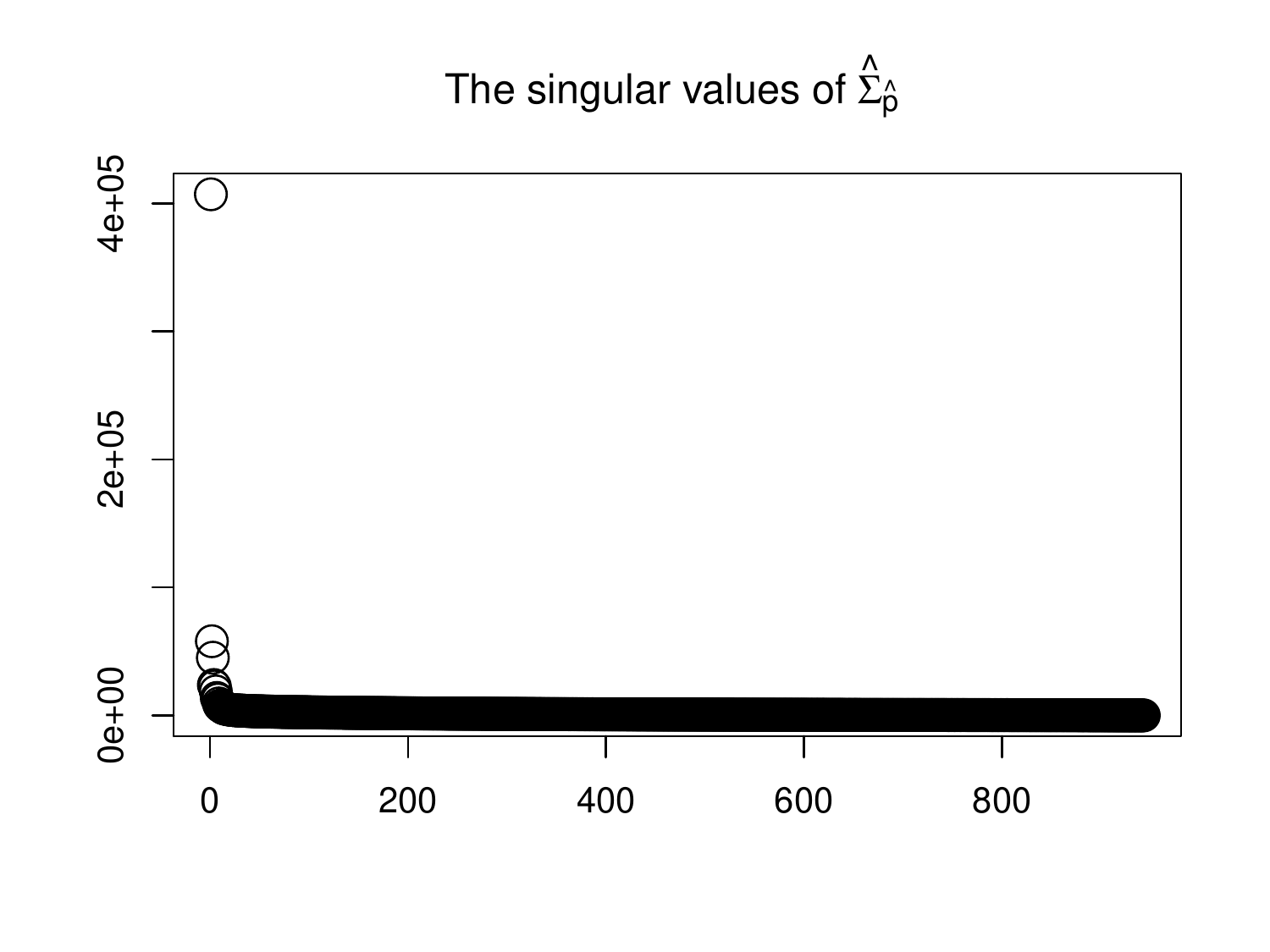}
\caption{The singular values of $\hat{\Sigma}_{\phat}$ computed by taking the MovieLens 100k data matrix as $\Mp$. 
	From this scree plot, we choose $\hat\rank$ to be 3.}
\label{fig:figure4}
\end{figure}
Since there exists a singular value gap between the 3rd and 4th singular values, we chose $\hat\rank=3$.
Then, we computed the estimators of the singular vectors and values and the estimator of the full low rank matrix as illustrated in Algorithms 
\ref{alg:UhatVhatLamhat} and \ref{alg:Mhat}.

The estimated singular vectors help us understand what the main factors of movie preferences are.
Table \ref{tbl1} shows lists of movies that characterize the top 3 singular vectors (factors of movie preferences).
Particularly, it presents 5 movies that correspond to the largest values in each singular vector and 5 movies that correspond to the smallest values. 
\begin{table}[!ht]
\centering
\caption{Lists of movies that characterize each of the top 3 singular vectors}
\label{tbl1}
\begin{tabular}{|c|c|c|}
\hline
\multirow{2}{*}{\begin{tabular}[c]{@{}c@{}}1st \\ singular\\ vector\end{tabular}} & \begin{tabular}[c]{@{}c@{}}One side \\  (well-known, top-rated) \end{tabular}  & \begin{tabular}[c]{@{}c@{}}Silence of the Lambs, Fargo, Star Wars,  \\ Return of the Jedi, Raiders of the Lost Ark \end{tabular} \\ \cline{2-3} 
                                             & \begin{tabular}[c]{@{}c@{}}The other side \\  (unknown, pooly-rated) \end{tabular} & \begin{tabular}[c]{@{}c@{}}A Further Gesture, Mat i syn, \\ A Very Natural Thing, Hush, Office Killer\end{tabular}              \\ \hline
\multirow{2}{*}{\begin{tabular}[c]{@{}c@{}}2nd \\ singular\\ vector\end{tabular}}  & \begin{tabular}[c]{@{}c@{}}One side \\  (box-office hit in 90's) \end{tabular} & \begin{tabular}[c]{@{}c@{}}Scream, Air Force One, The Rock, \\ Contact, Liar Liar\end{tabular}                                  \\ \cline{2-3} 
                                             & \begin{tabular}[c]{@{}c@{}}The other side \\  (classic in 40's-60's) \end{tabular} & \begin{tabular}[c]{@{}c@{}}Citizen Kane, The Graduate, Casablanca, \\ The African Queen, Dr. Strangelove\end{tabular}           \\ \hline
\multirow{2}{*}{\begin{tabular}[c]{@{}c@{}}3rd \\ singular\\ vector\end{tabular}}  & \begin{tabular}[c]{@{}c@{}}One side \\  (action, thriller) \end{tabular} & \begin{tabular}[c]{@{}c@{}}Jurassic Park, Top Gun, Speed, True Lies, \\ Batman\end{tabular}                                     \\ \cline{2-3} 
                                             & \begin{tabular}[c]{@{}c@{}}The other side \\  (drama) \end{tabular} & \begin{tabular}[c]{@{}c@{}}Il Postino, Secrets \& Lies, English Patient, \\ Full Monty, L.A. Confidential\end{tabular}          \\ \hline
\end{tabular}
\end{table}
The 1st factor has well-known and top-rated movies on one side and unknown and poorly-rated movies on the other side.
The 2nd factor has box-office hit movies in 1990's on one side and memorable classic movies in 1940's-1960's on the other side.
The 3rd factor has action and thriller movies on one side and quieter and drama movies on the other side.

The estimated singular values help us see how influential the main factors of movie preferences are.
Particularly, Figure \ref{fig:figure5} shows the estimated singular values and their $95\%$ confidence intervals.
For the standard deviation used in the confidence intervals, 
	we used $\Upsilon_{ii}^{-1/2}$ from \eqref{asympVar} in Theorem \ref{thm3}.
Computing $\Upsilon_{ii}^{-1/2}$ requires information on the values of the parameters $\M,\U,\V,\lam_i,\p$, and $\sig^2$,
	but we replaced these with the estimated values $\Mhat(\shat),\Uhat,\Vhat,\lamhat_i,\phat$, and $\tauhatphat/\n\phat^2$.
From Figure \ref{fig:figure5}, we observe that all 3 factors of movie preferences are significant.
\begin{figure}[!ht]
\centering
\includegraphics[width=4in]{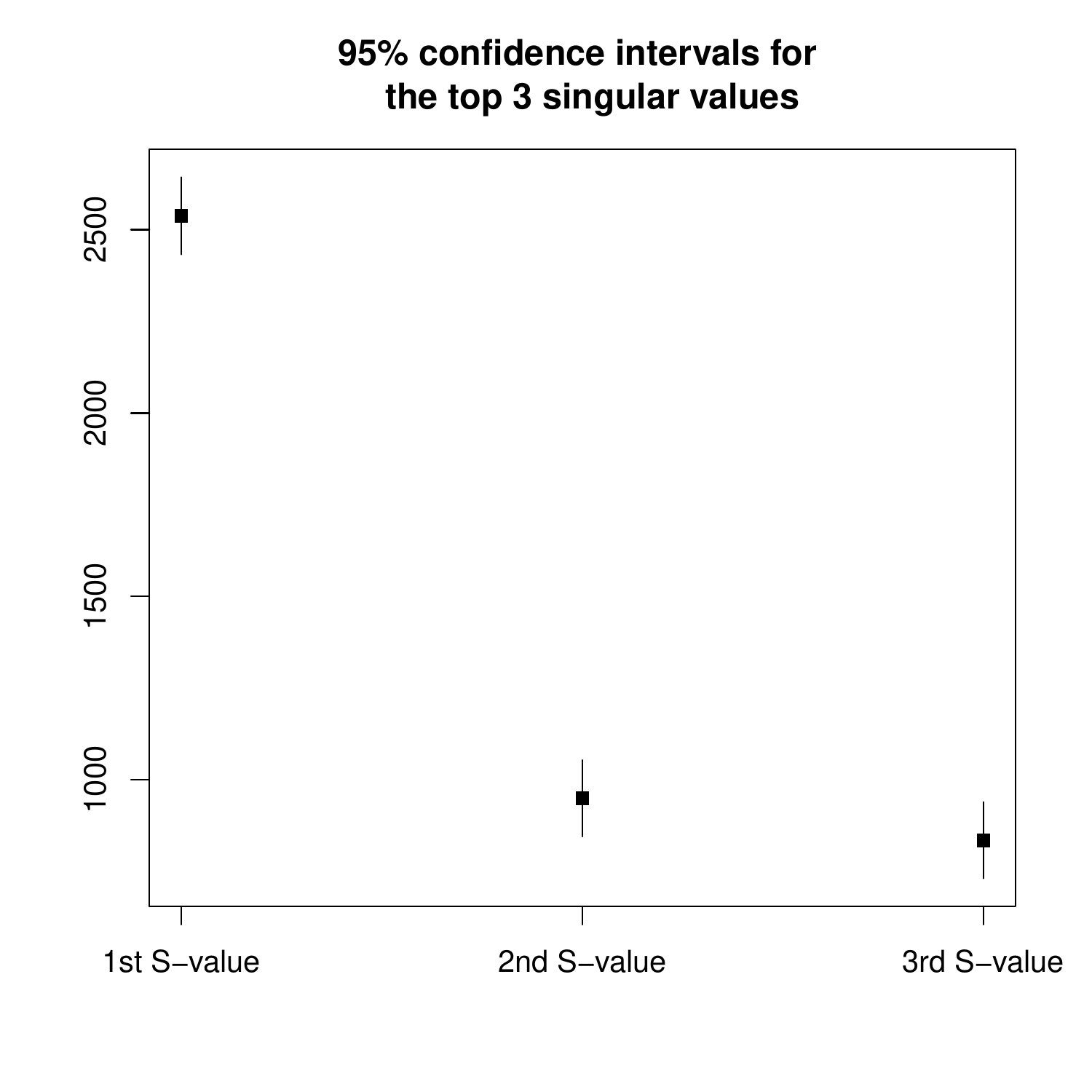}
\caption{The 3 estimated singular values and their $95\%$ confidence intervals.}
\label{fig:figure5}
\end{figure}

To find the RMSE of our estimator of the full low rank matrix, $\Mhat(\shat)$, 
	we used 5 training and 5 test data sets from 5-fold cross validation which is publicly provided in \cite{movielens100k}.
The RMSE was computed by 
$$\sqrt{\frac{\smallnorm{\mathcal{P}_{\Omega_{test}}(\Mhat(\shat))-\mathcal{P}_{\Omega_{test}}(\Mp)}_F^2}{|\Omega_{test}|}},$$
	where $\Omega_{test}$ contains indices of observed entries that belong to the test set, 
	$\mathcal{P}_{\Omega_{test}}$ for a matrix $A\in\real^{\n\times\d}$ denotes the projection of $A$ onto $\Omega_{test}$, and 
	$|\Omega_{test}|$ denotes the cardinality of $\Omega_{test}$. 
The average of the resulting RMSEs was 1.656.


\newpage
\section{Proofs} \label{proofs}


\subsection{Proofs for Theorem \ref{thm1}} \label{proofs:thm1}

The proof of the following proposition and  lemmas are in Appendix \ref{apdx1}.


\begin{proposition} \label{prop1}
Under the model setup in Section \ref{setup} and Assumption \ref{assume1}, we have for large $\n$ and $\d$,
\begin{eqnarray} \label{prop1:res1} 
\expect\norm{\sin \big( \V_\p^{(\m)}, \V^{(\m)}\big)}_F^2 &\le& \frac{C_1 \, \n^{-1}}{\p\,({b}_{\m}^2 - {b}_{\m+1}^2)^2}, \text{ and} 
\end{eqnarray}
\begin{eqnarray*}
\expect\norm{\sin \big( \U_\p^{(\m)}, \U^{(\m)}\big)}_F^2 &\le& \frac{C_2 \, \d^{-1}}{\p\,({b}_{\m}^2 - {b}_{\m+1}^2)^2},
\end{eqnarray*}
where $\V_\p$ and $\U_\p$ are defined in \eqref{Sigpp} and $C_1$ and $C_2$ are generic constants free of $\n,\d$, and $\p$.
\end{proposition}


\begin{lemma} \label{lem6}
Under the model setup in Section \ref{setup} and Assumption \ref{assume1}, for any given $\mu_1>0$, there exists a large constant $C_{\mu_1}>0$ such that
\begin{equation} \label{lem6:res1} 
	\frac{1}{\n\d} \norm{ \hat{\Sig}_\p-\expect\hat{\Sig}_\p } _2 \le 
	C_{\mu_1}\, \max \left\{ \p \frac{\log\n}{\d},\, \p^{3/2}\sqrt{\frac{\log \n}{\n}}
\right\}
\end{equation}
with probability at least $1-O\(\n^{-\mu_1}\)$, where $\hat{\Sig}_{\p}$ is defined in \eqref{Sigpp}.
Similarly, for any given $\mu_2>0$, there exists a large constant $C_{\mu_2}>0$ such that
\begin{equation*}
	\frac{1}{\n\d} \norm{ \hat{\Sig}_{\p t} - \expect \(\hat{\Sig}_{\p t}\) } _2 \le 
	C_{\mu_2}\, \max \left\{ \p \frac{\log\n}{\d},\, \p^{3/2}\sqrt{\frac{\log \n}{\d}}
\right\}
\end{equation*}
with probability at least $1-O\(\n^{-\mu_2}\)$, where $\hat{\Sig}_{\p t}$ is defined in \eqref{Sigpp}.
\end{lemma}


\begin{lemma} \label{lem4}
Under the model setup in Section \ref{setup} and Assumption \ref{assume1}, 
for any given $\nu_1>0$, there exists a large constant $C_{\nu_1}>0$ such that
\begin{equation} \label{lem4:res1} 
	\frac{1}{\n\d} \norm{ \hat{\Sig}_{\phat} - \hat{\Sig}_\p } _2 \le C_{\nu_1}\,\p^{3/2}\sqrt{\frac{\log\n}{\n\d}}\,\frac{1}{\d}
\end{equation}
with probability at least $1-O\(\n^{-\nu_1}\)$, 
where $\hat{\Sig}_{\phat}$ and $\hat{\Sig}_{\p}$ are defined in \eqref{Sigpphat} and \eqref{Sigpp}, respectively.
Similarly, for any given $\nu_2>0$, there exists a large constant $C_{\nu_2}>0$ such that
\begin{equation*}
	\frac{1}{\n\d} \norm{ \hat{\Sig}_{\phat t} - \hat{\Sig}_{\p t} } _2 \le C_{\nu_2}\,\p^{3/2}\sqrt{\frac{\log\n}{\n\d}}\,\frac{1}{\n}
\end{equation*}
with probability at least $1-O\(\n^{-\nu_2}\)$, 
where $\hat{\Sig}_{\phat t}$ and $\hat{\Sig}_{\p t}$ are defined in \eqref{Sigpphat} and \eqref{Sigpp}, respectively.
\end{lemma}


\begin{lemma} \label{lem5}
Under the model setup in Section \ref{setup} and Assumption \ref{assume1}, we have for large $\n$ and $\d$,
\begin{eqnarray} \label{lem5:res1} 
	\expect \norm{ \frac{1}{\n\d} \( \hat{\Sig}_{\phat}-\hat{\Sig}_\p \) \V_\p^{(\m)} }_F^2 &\le& 
	C_1 \,\max \Bigg\{ \frac{\p^3(1-\p)}{\n\d^3}, \frac{\p^2(1-\p)}{\n^2\d^{5/2}} \Bigg\}
\end{eqnarray}
and
\begin{eqnarray*}
	\expect \norm{ \frac{1}{\n\d} \( \hat{\Sig}_{\phat t}-\hat{\Sig}_{\p t} \) \U_\p^{(\m)} }_F^2 &\le& 
	C_2\, \max \Bigg\{ \frac{\p^3(1-\p)}{\d\n^3}, \frac{\p^2(1-\p)}{\d^2\n^{5/2}} \Bigg\},
\end{eqnarray*}
where $\hat{\Sig}_{\phat}$ and $\hat{\Sig}_{\phat t}$ are defined in \eqref{Sigpphat}, 
	$\hat{\Sig}_{\p}$, $\hat{\Sig}_{\p t}$, $\V_\p$, and $\U_\p$ are defined in \eqref{Sigpp}, 
	and $C_1$ and $C_2$ are generic constants free of $\n,\d$, and $\p$.
\end{lemma}


\begin{proof} [Proof of Theorem \ref{thm1}]
We only prove \eqref{Vconsist} because \eqref{Uconsist} can be proved similarly.

By triangle inequality and Proposition \ref{prop1}, we have
\begin{eqnarray} \label{thm1:eq0}
&&\expect\smallnorm{\sin \big( \Vhat^{(\m)}, \V^{(\m)}\big)}_F^2
\cr
&&\le 4\, \expect\smallnorm{\sin \big( \Vhat^{(\m)}, \V_\p^{(\m)}\big)}_F^2
	+ 4 \expect\smallnorm{\sin \big( \V_\p^{(\m)}, \V^{(\m)}\big)}_F^2
\cr
&&\le 4\, \expect\smallnorm{\sin \big( \Vhat^{(\m)}, \V_\p^{(\m)}\big)}_F^2
	+ \frac{C\, \n^{-1}}{\p\,({b}_{\m}^2 - {b}_{\m+1}^2)^2}.
\end{eqnarray}
Now, consider $\expect\smallnorm{\sin \big( \Vhat^{(\m)}, \V_\p^{(\m)}\big)}_F^2$.
Let 
\begin{eqnarray*}
E_1 = \left\{\max_{1\le i\le\d} \; \frac{1}{\n\d} \big| {\lam_\p^2}_{i} - {\ddot{\lam_\p}^2}_{i} \big|< t_1 \right\}, 
\end{eqnarray*}
where $t_1 = C_{1}'\, \p \frac{\log\n}{\d} + C_1''\, \p^{3/2}\sqrt{\frac{\log \n}{\n}}$,
and
\begin{eqnarray*}
E_2 = \left\{\frac{1}{\n\d} | {\lam_\p^2}_{\m+1} - {\lam_{\phat}^2}_{\m+1} |< t_2 \right\}.
\end{eqnarray*}
where $t_2 = C_{2} \,\p^{3/2}\sqrt{ \frac{\log\n}{\n\d}}\frac{1}{\d}$.
Then, by Weyl's theorem (\cite{li1998one}), Lemma \ref{lem6}, and Lemma \ref{lem4}, 
we have for large constants $C_{1}', C_{1}''$, and $C_{2}$,
$$
\prob(E_1^c) \le \prob\(\frac{1}{\n\d} \norm{\hat{\Sig}_\p - \expect \hat{\Sig}_\p}_2 \ge t_1 \) = O\( \n^{-4} \) 
\text{ and}
$$
$$
\prob(E_2^c) \le \prob\(\frac{1}{\n\d} \norm{\hat{\Sig}_{\phat} - \hat{\Sig}_\p}_2 \ge t_2 \) = O\( \n^{-4} \).
$$
Thus, for large $\n$ and $\d$,
\begin{eqnarray} \label{thm1:eq1}
&&\expect\smallnorm{\sin \big( \Vhat^{(\m)}, \V_\p^{(\m)}\big)}_F^2
\cr
&&= \expect \bigg\{ \smallnorm{\sin \big( \Vhat^{(\m)}, \V_\p^{(\m)}\big)}_F^2 \; \I_{(E_1 \cap E_2)^c} \bigg\}
\cr
&&\quad\quad\quad\quad
	+ \expect \bigg\{ \smallnorm{\sin \big( \Vhat^{(\m)}, \V_\p^{(\m)}\big)}_F^2\; \I_{E_1 \cap E_2} \bigg\}
\cr
&&\le \m\, \Big\{ \expect \big( \I_{E_2^c} \big) 	+ \expect \big( \I_{E_1^c} \big) \Big\}
	+ \expect \Bigg\{ \frac{ \norm{ \frac{1}{\n\d} \( \hat{\Sig}_{\phat}-\hat{\Sig}_\p \) \V_\p^{(\m)} }_F^2}
		{ \( \frac{1}{\n\d} | {\lam_\p^2}_{\m} - {\lam_{\phat}^2}_{\m+1} | \)^2}
			\; \I_{E_1 \cap E_2} \Bigg\}
\cr
&&\le c \n^{-4}
	+ \expect \Bigg\{ \frac{ \norm{ \frac{1}{\n\d} \( \hat{\Sig}_{\phat}-\hat{\Sig}_\p \) \V_\p^{(\m)} }_F^2 \; \I_{E_1 \cap E_2}}
		{ \( \frac{1}{\n\d} | {\ddot{\lam_\p}^2}_{\m} - {\ddot{\lam_\p}^2}_{\m+1} | 
			-t_2 - 2 t_1 \)^2} \Bigg\}
\cr
&&\le c \n^{-4}
	+ \expect \Bigg\{ \frac{ \norm{ \frac{1}{\n\d} \( \hat{\Sig}_{\phat}-\hat{\Sig}_\p \) \V_\p^{(\m)} }_F^2}
		{ \( \frac{1}{2\n\d} | {\ddot{\lam_\p}^2}_{\m} - {\ddot{\lam_\p}^2}_{\m+1} | \)^2} \Bigg\}
\cr
&&\le c \n^{-4} + \frac{C(1-\p)}{({b}_{\m}^2 - {b}_{\m+1}^2)^2}\max \Bigg\{ \frac{1}{\p\n\d^3}, \frac{1}{\p^2\n^2\d^{5/2}} \Bigg\},
\end{eqnarray}
where $\I_E$ is an indicator function of an event $E$, the first inequality holds 
	by the fact that $\|\sin (\Vhat^{(\m)}, \V_\p^{(\m)})\|_F^2 \le \m$
	and Davis-Kahan $\sin \theta$ theorem (Theorem 3.1 in \cite{li1998two}),
and the last inequality is due to Lemma \ref{lem5}.
 
By \eqref{thm1:eq0} and \eqref{thm1:eq1}, the result \eqref{Vconsist} follows.
\end{proof}


\subsection{Proofs for Theorem \ref{thm2}} \label{proofs:thm2}

The proof of the following propositions are in Appendix \ref{apdx2}.


\begin{proposition} \label{lem8}
Under the assumptions in Theorem \ref{thm2}, we have  
\begin{eqnarray*}
&&\sqrt{\n\d} \,\Gamma_{\n\d}^{-1/2} \[
\begin{pmatrix}
\frac{1}{\n\d\,\p^2} \sum_{i=1}^{\m} {\lam_\p^2}_i \smallskip \\ 
\frac{\p^2}{\n\d} \sum_{i=1}^{\m}( \lam^2_i + \n\sig^2 ) \,\phat
\end{pmatrix}
-
\begin{pmatrix}
\frac{1}{\n\d} \sum_{i=1}^{\m}\left[ \lam^2_i + \n\sig^2 \right] \smallskip \\ 
\frac{\p^3}{\n\d} \sum_{i=1}^{\m}( \lam^2_i + \n\sig^2 ) 
\end{pmatrix}
\] \cr\cr
&&\rightarrow  \mathcal{N}\( 0, I_2\) \text{ in distribution,} \quad \text{as $\n,\d\to\infty$,}
\end{eqnarray*} 
where $\lam_{\p i}$, $\lam_{i}$, and $\phat$ are defined in \eqref{Sigpp}, \eqref{mod}, and \eqref{phat}, respectively, and $\Gamma_{\n\d} = \Gamma_{\n\d}^T \in \real^{2\times2}$ consists of 
\begin{eqnarray*}
&&(\Gamma_{\n\d})_{11} = \frac{4(1-\p)}{\p}\sum_{k=1}^\n \sum_{h=1}^\d {\M}_{kh}^2 
	\Bigg\{ \sum_{i=1}^\m b_i\U_{ik}\V_{ih} \Bigg\}^2 + \frac{4\sig^2}{\p} \sum_{i=1}^\m b_i^2,
\cr
&&(\Gamma_{\n\d})_{12} = 2\p^2(1-\p) \Bigg(\sum_{i=1}^\m  b^2_i\Bigg)^2,
	\text{ and } (\Gamma_{\n\d})_{22} = \p^5(1-\p) \(\sum_{i=1}^\m b_{i}^2\)^2.
\end{eqnarray*}

\end{proposition}


\begin{proposition} \label{prop5}
Under the model setup in Section \ref{setup} and Assumption \ref{assume1}, 
let 
$$\hat{\tau}_{\p} = \frac{1}{\d-\rank} \trace\(\V_{\p c}^T \hat\Sig_\p \V_{\p c}\),$$
where $\hat\Sig_\p$ and $\V_{\p c}$ are defined in \eqref{Sigpp}.
Then, we have
$\hat{\tau}_{\p} - \n\p^2\sig^2 = O_p\(\p\sqrt{\n}\)$.
\end{proposition}


\begin{proof}[Proof of Theorem \ref{thm2}]
We have
\begin{eqnarray*}
&&\frac{1}{\sqrt{\n\d}} \left\{ \sum_{i=1}^{\m} \lamhat_i^2 - \sum_{i=1}^{\m} \lam^2_i \right\}
\cr
&&=\frac{1}{\sqrt{\n\d}} \left\{ \( \phat^{-2} \sum_{i=1}^{\m} {\lam_{\phat}^2}_i - \sum_{i=1}^{\m}\left[ \lam^2_i + \n\sig^2 \right] \)
	+ \m \( \n\sig^2 - \frac{1}{\phat^{2}} \tauhatphat \)  \right\}
\cr
&&=\frac{1}{\sqrt{\n\d}} \left\{ \( a \) + \m \( b \)  \right\}.
\end{eqnarray*}

First, consider the term $(a)$.
We have
\begin{eqnarray} \label{prop3:eq1}
(a)
&=&\frac{1}{\phat^2} \, \trace\( \Vhat^{(\m)T} \hat{\Sig}_{\phat} \Vhat^{(\m)} \) - \sum_{i=1}^{\m}\left[ \lam^2_i + \n\sig^2 \right]
\cr 
&=&\bigg\{\frac{1}{\p^2} \, \trace\( \Vhat^{(\m)T} \hat{\Sig}_\p \Vhat^{(\m)} \) - \sum_{i=1}^{\m}\left[ \lam^2_i + \n\sig^2 \right]\bigg\}
\cr 
&&
	+ \bigg\{ \frac{1}{\p^2} \, \trace\( \Vhat^{(\m)T} \hat{\Sig}_{\phat} \Vhat^{(\m)} \) - \frac{1}{\p^2} \, \trace\( \Vhat^{(\m)T} \hat{\Sig}_\p \Vhat^{(\m)} \) \bigg\}
\cr 
&&
	+ \bigg\{ \frac{1}{\phat^2} \, \trace\( \Vhat^{(\m)T} \hat{\Sig}_{\phat} \Vhat^{(\m)} \) - \frac{1}{\p^2} \, \trace\( \Vhat^{(\m)T} \hat{\Sig}_{\phat} \Vhat^{(\m)} \) \bigg\}
\cr 
&=&(i) + (ii) + (iii) .
\end{eqnarray}

By \eqref{sineDist}, there is $\O \in \mathbb{V}_{\m,\m}$ such that 
$$\smallnorm{\Vhat^{(\m)} - \V_\p^{(\m)}\O}_F^2 \le 2\smallnorm{\sin (\Vhat^{(\m)}, \V_\p^{(\m)})}_F^2
	\quad\text{and}\quad
	\O_i^T\V_\p^{(\m)T}\hat{\Sig}_\p\V_\p^{(\m)}\O_i=\lam_{\p i}^2,$$
where $\O_i$ is the $i$-th column of $\O$.
Then, the term $(i)$ is
\begin{eqnarray} \label{prop3:eq2(i)}
(i)
&=&\frac{1}{\p^2} \, \trace\( \O^T\V_\p^{(\m)T} \hat{\Sig}_\p \V_\p^{(\m)}\O \) - \sum_{i=1}^{\m}\left[ \lam^2_i + \n\sig^2 \right]
\cr
&&\quad\quad
	+\frac{1}{\p^2} \, \trace\( \Vhat^{(\m)T} \hat{\Sig}_\p \Vhat^{(\m)} - \O^T\V_\p^{(\m)T} \hat{\Sig}_\p \V_\p^{(\m)}\O \)
\cr
&=& \frac{1}{\p^2} \, \trace\( \V_\p^{(\m)T} \hat{\Sig}_\p \V_\p^{(\m)} \) - \sum_{i=1}^{\m}\left[ \lam^2_i + \n\sig^2 \right]
\cr
&&\quad\quad
	+\frac{1}{\p^2} \, \sum_{i=1}^\m \( {\Vhat}_i^{T} \hat{\Sig}_\p {\Vhat}_i - \O_i^T{\V_\p}^{T} \hat{\Sig}_\p {\V_\p} \O_i \) 
\cr
&=& \frac{1}{\p^2} \sum_{i=1}^{\m} {\lam_\p^2}_i - \sum_{i=1}^{\m}\left[ \lam^2_i + \n\sig^2 \right]
	+ O_p\(\frac{1}{\p\d^2}\),
\end{eqnarray}
where the last equality holds by the fact that 
\begin{eqnarray}  \label{prop3:eq2(i)-supp}
&&\left| \sum_{i=1}^\m 
	\(\Vhat_i^{T}\hat\Sig_\p \Vhat_i - \O_i^T \V_\p^{(\m)T}\hat\Sig_\p \V_\p^{(\m)}\O_i \) \right|
\cr
&&=\left| \sum_{i=1}^\m \[(\Vhat_i-\V_\p^{(\m)}\O_i)^T\hat\Sig_\p(\Vhat_i-\V_\p^{(\m)}\O_i) 
	+ 2 \lam_{\p i}^2 \O_i^T \V_\p^{(\m)T}\Vhat_i - 2 \lam_{\p i}^2 \] \right|
\cr
&&=\left| \sum_{i=1}^\m \[(\Vhat_i-\V_\p^{(\m)}\O_i)^T\hat\Sig_\p(\Vhat_i-\V_\p^{(\m)}\O_i) 
	- \lam_{\p i}^2 \norm{ \Vhat_i-\V_\p^{(\m)}\O_i }_2^2 \] \right|
\cr
&&\le2 \lam_{\p 1}^2 \sum_{i=1}^\m \norm{ \Vhat_i-\V_\p^{(\m)}\O_i }_2^2
\cr
&&=2 \lam_{\p 1}^2 \norm{ \Vhat^{(\m)}-\V_\p^{(\m)}\O }_F^2
\cr
&&=O_p\(\frac{\p}{\d^2}\),
\end{eqnarray}
where the last equality is due to \eqref{sineDist}, \eqref{thm1:eq1}, and \eqref{prop3:eq2(i)-supp-supp} below;
by the application of Weyl's theorem (\cite{li1998one}) and Lemma \ref{lem6},
	we can show
\begin{eqnarray} \label{prop3:eq2(i)-supp-supp}
\lam_{\p 1}^2 = O_p(\p^2\n\d).
\end{eqnarray}

The term $(ii)$ is
\begin{eqnarray} \label{prop3:eq2(ii)}
\expect \left| (ii) \right| 
&=&\expect \left| \frac{1}{\p^2} \, (\phat - \p)\, \trace\( \Vhat^{(\m)T} \diag(\hat{\Sig}) \Vhat^{(\m)} \) \right| 
\cr
&\le&\frac{\m}{\p^2} \; \expect \left| (\phat - \p) \; \max_{1\le i\le \m} {\Vhat}_i^T \diag(\hat{\Sig}) {\Vhat}_i \right| 
\cr
&\le& \frac{\m}{\p^2} \bigg\{ \expect(\phat - \p)^2 \bigg\}^{1/2} 
	\bigg\{ \expect \Big[ \max_{1\le i\le \m} {\Vhat}_i^T \diag(\hat{\Sig}) {\Vhat}_i \Big]^2 \bigg\}^{1/2}
\cr
&\le& \frac{\m}{\p^2} \bigg\{ \expect(\phat - \p)^2 \bigg\}^{1/2} 
	\bigg\{ \expect \Big[ \norm{ \diag(\hat{\Sig}) }_2^2 \Big] \bigg\}^{1/2}
\cr
&=& \frac{\m}{\p^2} \sqrt{ \frac{\p(1 - \p)}{\n\d} } \,\bigg\{ \expect \Big[ \norm{ \diag(\hat{\Sig}) }_2^2 \Big] \bigg\}^{1/2}
\cr
&=&O\(\max\left\{\frac{1}{\,\p\,}, \sqrt{\frac{\n}{\p\d}} \right\}\),
\end{eqnarray}
where the second inequality is due to H\"older's inequality and the last equality holds by the fact that
\begin{eqnarray*}
&&\expect \Big[ \smallnorm{ \diag(\hat{\Sig}) }_2^2 \Big]
\cr
&&\le 4\, \expect \Big[ \smallnorm{ \diag(\hat{\Sig}) - \p\,\diag(\M^T\M) - \n\p\sig^2I_\d }_2^2 
\cr
&&\quad\quad\quad\quad\quad\quad\quad\quad
	+ \smallnorm{ \p\,\diag(\M^T\M) + \n\p\sig^2 I_\d }_2^2 \Big]
\cr
&&=4\, \expect \Bigg[ \max_{1\le h \le \d} 
		\Big| \sum_{k=1}^\n \big(\Mp_{kh}^2 - \p{\M}_{kh}^2 - \p\sig^2 \big) \Big|^2 \Bigg]
\cr
&&\quad\quad\quad\quad\quad\quad\quad\quad
	+ 4\, \Big\{\max_{1\le h\le\d} \p \sum_{k=1}^\n {\M}_{kh}^2 + \n\p\sig^2 \Big\}^2
\cr
&&\le 4\,\sum_{h=1}^\d \expect \Bigg\{ \bigg| \sum_{k=1}^\n \big[\Mp_{kh}^2 - \p ({\M}_{kh}^2 + \sig^2) \big] \bigg|^2 \Bigg\}
	+ 4\,\Big\{ \n\p (\L^2 + \sig^2) \Big\}^2
\cr
&&=4\,\sum_{h=1}^\d \sum_{k=1}^\n \expect \big[\Mp_{kh}^2 - \p ({\M}_{kh}^2 + \sig^2) \big]^2
	+ 4\,\Big\{ \n\p (\L^2 + \sig^2) \Big\}^2
\cr
&&= O\(\max\{\p\n\d, \p^2\n^2\}\).
\end{eqnarray*}

The term $(iii)$ in \eqref{prop3:eq1} is
\begin{eqnarray} \label{prop3:eq2(iii)}
(iii)
&=&\(\frac{1}{\phat^2} - \frac{1}{\p^2} \) \left[ \trace\( \Vhat^{(\m)T} \hat{\Sig}_{\phat} \Vhat^{(\m)} \) - \p^2 \sum_{i=1}^{\m}\( \lam^2_i + \n\sig^2 \) \right] 
\cr
&&
	+ \(\frac{1}{\phat^2} - \frac{1}{\p^2} \) \p^2 \sum_{i=1}^{\m}\( \lam^2_i + \n\sig^2 \)
\cr
&=&\(\frac{1}{\phat^2} - \frac{1}{\p^2} \) \Bigg[ 
\trace\( \Vhat^{(\m)T} \hat{\Sig}_{\phat} \Vhat^{(\m)} \) - \trace\( \Vhat^{(\m)T} \hat{\Sig}_\p \Vhat^{(\m)} \)
\cr
&&\quad\quad\quad\quad\quad\quad\quad
+ \trace\( \Vhat^{(\m)T} \hat{\Sig}_\p \Vhat^{(\m)} \) - \trace\( \O^T\V_\p^{(\m)T} \hat{\Sig}_\p \V_\p^{(\m)}\O \)
\cr
&&\quad\quad\quad\quad\quad\quad\quad
+ \trace\( \O^T\V_\p^{(\m)T} \hat{\Sig}_\p \V_\p^{(\m)}\O \) - \p^2 \, \sum_{i=1}^{\m}\( \lam^2_i + \n\sig^2 \) \Bigg] 
\cr 
&&
	+ \(\frac{1}{\phat^2} - \frac{1}{\p^2} \) \, \p^2 \, \sum_{i=1}^{\m}\( \lam^2_i + \n\sig^2 \) 
\cr
&=&O_p\(\frac{1}{\sqrt{\p^5\n\d}} \) \,\Bigg[ 
O_p\(\max\left\{\p, \sqrt{\frac{\p^3\n}{\d}} \right\}\) + O_p\(\frac{\p}{\d^2}\) + O_p\(\sqrt{\p^3\n\d}\) \Bigg] 
\cr 
&&\quad\quad
	+ \(\frac{1}{\phat^2} - \frac{1}{\p^2} \) \, \p^2 \, \sum_{i=1}^{\m}\( \lam^2_i + \n\sig^2 \) 
\cr
&=&O_p\(\frac{1}{\p}\)
	+ \(\frac{1}{\phat^2} - \frac{1}{\p^2} \) \, \p^2 \sum_{i=1}^{\m}\( \lam^2_i + \n\sig^2 \) ,
\end{eqnarray}
where  the third equality is due to \eqref{prop3:eq2(ii)}, \eqref{prop3:eq2(i)-supp}, Proposition \ref{lem8}, and the fact that
\begin{equation} \label{prop3:eq2(iii)-1}
\sqrt{\n\d} \(\frac{1}{\phat^2} - \frac{1}{\p^2}\) \; \rightarrow \; \mathcal{N}\(0, \frac{4(1-\p)}{\p^5}\) \text{ in distribution,} \text{ as } \n,\d\to\infty, 
\end{equation}
by CLT and Delta method.
From \eqref{prop3:eq2(i)}, \eqref{prop3:eq2(ii)}, and \eqref{prop3:eq2(iii)}, we have
\begin{eqnarray} \label{prop3(a)}
&&(a) = \frac{1}{\p^2} \sum_{i=1}^{\m} {\lam_\p^2}_i - \sum_{i=1}^{\m}\left[ \lam^2_i + \n\sig^2 \right]
\cr
&&\quad\quad\quad\quad
	+ \(\frac{1}{\phat^2} - \frac{1}{\p^2} \) \, \p^2 \sum_{i=1}^{\m}\( \lam^2_i + \n\sig^2 \)
	+ o_p\(\sqrt{\frac{\n\d}{\p}}\).
\end{eqnarray}

Second, the term $(b)$ is 
\begin{eqnarray} \label{prop3(b)}
(b) &=& \n\sig^2 - \frac{1}{\phat^2}\tauhatphat
\cr
&=& \(\n\sig^2 - \frac{1}{\p^2}\tauhatp\) +  \(\frac{1}{\p^2}-\frac{1}{\phat^2}\)\tauhatphat  +  \frac{1}{\p^2}\(\tauhatp - \tauhatphat\)
\cr
&=& O_p\(\frac{\sqrt{\n}}{\,\p}\) + O_p\(\sqrt{\frac{\n}{\p\d}}\) +  \frac{1}{\p^2}\(\tauhatp - \tauhatphat\)
\cr
&=& o_p\(\sqrt{\frac{\n\d}{\p}}\),
\end{eqnarray}
where the third equality is due to Proposition \ref{prop5} and \eqref{prop3:eq2(iii)-1}, and 
	the last equality holds by the fact that 
	there is $\tilde\O \in \mathbb{V}_{\d-\rank,\d-\rank}$ by \eqref{sineDist} such that 
$$\smallnorm{\Vhat_c^{(\m)} - \V_{\p c}^{(\m)}\tilde\O}_F^2 \le 2\smallnorm{\sin (\Vhat_c^{(\m)}, \V_{\p c}^{(\m)})}_F^2
	\quad\text{and}\quad
	\tilde\O_i^T\V_{\p c}^T\hat{\Sig}_\p\V_{\p c}\tilde\O_i=\lam_{\p\, \rank+i}^2,$$
where $\tilde\O_i$ is the $i$-th column of $\tilde\O$,
and that
\begin{eqnarray*}
&&|\tauhatp - \tauhatphat |
\cr
&&=\frac{1}{(\d-\rank)} \Big| \trace\(\tilde\O^T\V_{\p c}^T \hat{\Sig}_\p \V_{\p c}\tilde\O\) - \trace\(\Vhatc^T \hat{\Sig}_\p \Vhatc\) 
\cr
&&\quad\quad\quad\quad\quad\quad\quad\quad
		+ \trace\(\Vhatc^T \hat{\Sig}_\p \Vhatc\) - \trace\(\Vhatc^T \hat{\Sig}_{\phat} \Vhatc\) \Big|
\cr
&&\le \frac{1}{(\d-\rank)} 
	\left| \trace\(\tilde\O^T\V_{\p c}^T \hat{\Sig}_\p \V_{\p c}\tilde\O\) - \trace\(\Vhatc^T \hat{\Sig}_\p \Vhatc\) \right|
\cr
&&\quad\quad\quad\quad\quad\quad\quad\quad
	+ \frac{1}{(\d-\rank)} 
		\left| \trace\(\Vhatc^T \hat{\Sig}_\p \Vhatc\) - \trace\(\Vhatc^T \hat{\Sig}_{\phat} \Vhatc\) \right|
\cr
&&\le \frac{1}{(\d-\rank)} \,4 {\lam_\p^2}_1\norm{\sin (\V_{\p c}, \Vhat_c)}_F^2
	+ \frac{1}{(\d-\rank)} 
		\left| (\phat - \p) \trace\(\Vhatc^T \diag( \hat{\Sig} ) \Vhatc\) \right|
\cr
&&= \frac{1}{(\d-\rank)} \,4 {\lam_\p^2}_1\norm{\sin (\V_\p, \Vhat)}_F^2
	 + O_p\(\max \left\{ \p,\,\p^{3/2}\sqrt{\frac{\n}{\d}} \right\}\)
\cr
&&= O_p\(\frac{\p}{\d^3}\) + O_p\(\max \left\{ \p,\,\p^{3/2}\sqrt{\frac{\n}{\d}} \right\}\),
\end{eqnarray*}
where the second inequality can be derived similarly to \eqref{prop3:eq2(i)-supp},
	the second equality holds similarly to \eqref{prop3:eq2(ii)}, and 
	the last equality is due to \eqref{thm1:eq1} and \eqref{prop3:eq2(i)-supp-supp}.

Combining the results in \eqref{prop3(a)} and \eqref{prop3(b)}, we have
\begin{eqnarray*}
&&\frac{1}{\sqrt{\n\d}} \left\{ \sum_{i=1}^{\m} \lamhat_i^2 - \sum_{i=1}^{\m} \lam^2_i \right\}
\cr
&&=\frac{1}{\sqrt{\n\d}} \left\{ \( a \) + \m \( b \)  \right\}
\cr
&&= \frac{1}{\sqrt{\n\d}} \left\{ \frac{1}{\p^2} \sum_{i=1}^{\m} {\lam_\p^2}_i - \sum_{i=1}^{\m}\left[ \lam^2_i + \n\sig^2 \right]
	+ \(\frac{1}{\phat^2} - \frac{1}{\p^2} \) \, \p^2 \sum_{i=1}^{\m}\( \lam^2_i + \n\sig^2 \) \right\} 
\cr
\cr
&&\quad\quad
	+o_p(1).
\end{eqnarray*}
Thus, by Proposition \ref{lem8}, Delta method and Slutsky's theorem, we have
\begin{eqnarray*}
&&\frac{1}{\sqrt{\n\d} \sigma_\lambda} \left\{ \sum_{i=1}^{\m} \lamhat_i^2 - \sum_{i=1}^{\m} \lam^2_i \right\}
\; \rightarrow \; \mathcal{N}\(0, 1\) \text{ in distribution,} \quad\text{as }\n,\d\to\infty,
\end{eqnarray*}
where 
$\sigma_\lambda^2
=\(1 \;\; -2\p^{-3}\) \Gamma_{\n\d} \begin{pmatrix}
1\\ 
-2\p^{-3}
\end{pmatrix}
$.
\end{proof}


\subsection{Proofs for Theorem \ref{thm3}} \label{proofs:thm3}

The proof of the following Proposition is in Appendix \ref{apdx3}.


\begin{proposition} \label{prop4}
Under the model setup in Section \ref{setup}, Assumption \ref{assume1}, and Assumption \ref{assume2}(2), we have
$$
\norm{\Mhat(\s_0)-\M}_F^2 = \frac{1}{\p\,{b}_{\rank}^4} \, O_p\(\n\),
$$
where $\Mhat(\s_0)$ are defined in \eqref{s0} and \eqref{Mhats} and $\M$ is defined in \eqref{mod}.
\end{proposition}


\begin{proof}[Proof of Theorem \ref{thm3}]
For any given $\eta >0$, we have for a large $\n$,
\begin{eqnarray*}
&&\prob\(\min_{\s \in \{-1,1\}^\rank} \; \smallnorm{ \P(\Mhat(\s)) - \P(\Mp) }_F^2 < \smallnorm{ \P(\Mhat(\s_0)) - \P(\Mp) }_F^2 \) 
	\le \eta/2
\end{eqnarray*}
by Assumption \ref{assume2}(1).
Also, for any given $\eta >0$, we can find $C_\eta>0$, free of $\n$, $\d$, and $\p$, such that for large $\n$,
\begin{eqnarray*}
\prob \( \frac{\p\,{b}_{\rank}^4}{\n} \norm{\Mhat(\s_0)-\M}_F^2 \ge C_\eta \)	\le \eta/2
\end{eqnarray*}
by Proposition \ref{prop4}.
Therefore, for any given $\eta >0$, we can find $C_\eta>0$ such that
\begin{eqnarray*}
&&\prob \( \frac{\p\,{b}_{\rank}^4}{\n} \norm{\Mhat(\shat)-\M}_F^2 \ge C_\eta \) 
\cr
&&= \prob \( \frac{\p\,{b}_{\rank}^4}{\n} \norm{\Mhat(\s_0)-\M}_F^2 \ge C_\eta , \s_0=\shat \) 
\cr
&&\quad
	+ \prob \( \frac{\p\,{b}_{\rank}^4}{\n} \norm{\Mhat(\shat)-\M}_F^2 \ge C_\eta , \s_0\ne\shat \) 
\cr
&&\le \prob \( \frac{\p\,{b}_{\rank}^4}{\n} \norm{\Mhat(\s_0)-\M}_F^2 \ge C_\eta \) 
\cr
&&\quad
	+ \prob\(\min_{\s \in \{-1,1\}^\rank} \; \smallnorm{ \P(\Mhat(\s)) - \P(\Mp) }_F^2 < \smallnorm{ \P(\Mhat(\s_0)) - \P(\Mp) }_F^2 \) 
\cr
&&\le \eta/2 + \eta/2
\cr
&&= \eta.
\end{eqnarray*}

Or, for any given $\eta>0$ and $\zeta>0$, there exists $N_\zeta>0$ such that for all $\n\ge N_\zeta$, 
\begin{eqnarray*}
&&\prob \( \frac{\p\,{b}_{\rank}^4}{h_\n \n} \norm{\Mhat(\shat)-\M}_F^2 > \eta \) 
\cr
&&= \prob \( \frac{\p\,{b}_{\rank}^4}{h_\n \n} \norm{\Mhat(\s_0)-\M}_F^2 > \eta , \s_0=\shat \) 
\cr
&&\quad
	+ \prob \( \frac{\p\,{b}_{\rank}^4}{h_\n \n} \norm{\Mhat(\shat)-\M}_F^2 > \eta , \s_0\ne\shat \) 
\cr
&&\le \prob \( \frac{\p\,{b}_{\rank}^4}{h_\n \n} \norm{\Mhat(\s_0)-\M}_F^2 \ge \eta \) 
\cr
&&\quad
	+ \prob\(\min_{\s \in \{-1,1\}^\rank} \; \smallnorm{ \P(\Mhat(\s)) - \P(\Mp) }_F^2 < \smallnorm{ \P(\Mhat(\s_0)) - \P(\Mp) }_F^2 \) 
\cr
&&\le \zeta/2 + \zeta/2
\cr
&&= \zeta,
\end{eqnarray*}
where the second inequality holds due to Assumption \ref{assume2}(1) and Proposition \ref{prop4}.
\end{proof}


  \bibliography{references}

{\large \bf References}
\begin{thebibliography}{11}
\expandafter\ifx\csname
natexlab\endcsname\relax\def\natexlab#1{#1}\fi
\expandafter\ifx\csname url\endcsname\relax
  \def\url#1{\texttt{#1}}\fi
\expandafter\ifx\csname urlprefix\endcsname\relax\def\urlprefix{URL
}\fi
\bibitem[Anderson et~al.(1958)Anderson, Anderson, Anderson, and Anderson]{anderson1958}
T.W.~Anderson.
\newblock \emph{An introduction to multivariate statistical analysis},
\newblock  volume 2. Wiley New York, 1958.

\bibitem[Bennett and Lanning(2007)]{bennett2007}
J.~Bennett and S.~Lanning.
\newblock The netflix prize.
\newblock In \emph{Proceedings of KDD cup and workshop}, volume 2007, page~35,
  2007.

\bibitem[Cai et~al.(2010)Cai, Cand{\`e}s, and Shen]{cai2010}
J.-F. Cai, E.~J. Cand{\`e}s, and Z.~Shen.
\newblock A singular value thresholding algorithm for matrix completion.
\newblock \emph{SIAM Journal on Optimization}, 20\penalty0 (4):\penalty0
  1956--1982, 2010.

\bibitem[Cai and Zhou(2013)]{cai2013}
T.~T. Cai and W.-X. Zhou.
\newblock Matrix completion via max-norm constrained optimization.
\newblock \emph{arXiv preprint arXiv:1303.0341}, 2013.

\bibitem[Cand{\`e}s and Plan(2010)]{candes2010plan}
E.~J. Cand{\`e}s and Y.~Plan.
\newblock Matrix completion with noise.
\newblock \emph{Proceedings of the IEEE}, 98\penalty0 (6):\penalty0 925--936,
  2010.

\bibitem[Cand{\`e}s and Plan(2011)]{candes2011tight}
E.~J. Cand{\`e}s and Y.~Plan.
\newblock Tight oracle inequalities for low-rank matrix recovery from a minimal
  number of noisy random measurements.
\newblock \emph{Information Theory, IEEE Transactions on}, 57\penalty0
  (4):\penalty0 2342--2359, 2011.

\bibitem[Cand{\`e}s and Recht(2009)]{candes2009recht}
E.~J. Cand{\`e}s and B.~Recht.
\newblock Exact matrix completion via convex optimization.
\newblock \emph{Foundations of Computational mathematics}, 9\penalty0
  (6):\penalty0 717--772, 2009.

\bibitem[Chatterjee(2014)]{chatterjee2014}
S.~Chatterjee.
\newblock Matrix estimation by universal singular value thresholding.
\newblock \emph{The Annals of Statistics}, 43\penalty0 (1):\penalty0 177--214,
  2014.

\bibitem[Cho et~al.(2015)Cho, Kim, and Rohe]{cho2015nips}
J.~Cho, D.~Kim, and K.~Rohe.
\newblock Intelligent initialization and adaptive thresholding for iterative
  matrix completion; some statistical and algorithmic theory for
  adaptive-impute, 2015.

\bibitem[Davenport et~al.(2014)Davenport, Plan, van~den Berg, and
  Wootters]{davenport2014}
M.~A. Davenport, Y.~Plan, E.~van~den Berg, and M.~Wootters.
\newblock 1-bit matrix completion.
\newblock \emph{Information and Inference}, 3\penalty0 (3):\penalty0 189--223,
  2014.

\bibitem[Fan et~al.(2013)Fan, Liao, and Mincheva]{fan2013}
J.~Fan, Y.~Liao, and M.~Mincheva.
\newblock Large covariance estimation by thresholding principal orthogonal
  complements.
\newblock \emph{Journal of the Royal Statistical Society: Series B (Statistical
  Methodology)}, 75\penalty0 (4):\penalty0 603--680, 2013.

\bibitem[Fazel(2002)]{fazel2002}
M.~Fazel.
\newblock \emph{Matrix rank minimization with applications}.
\newblock PhD thesis, PhD thesis, Stanford University, 2002.

\bibitem[\protect\citeauthoryear{GroupLens}{GroupLens}{2015}]{movielens100k}
GroupLens (2015).
\newblock Movielens100k {@MISC}.
\newblock \url{http://grouplens.org/datasets/movielens/}.

\bibitem[Gross(2011)]{gross2011}
D.~Gross.
\newblock Recovering low-rank matrices from few coefficients in any basis.
\newblock \emph{Information Theory, IEEE Transactions on}, 57\penalty0
  (3):\penalty0 1548--1566, 2011.

\bibitem[Hastie et~al.(2014)Hastie, Mazumder, Lee, and Zadeh]{hastie2014}
T.~Hastie, R.~Mazumder, J.~Lee, and R.~Zadeh.
\newblock Matrix completion and low-rank svd via fast alternating least
  squares.
\newblock \emph{arXiv preprint arXiv:1410.2596}, 2014.

\bibitem[Keshavan et~al.(2009)Keshavan, Montanari, and Oh]{keshavan2009}
R.~Keshavan, A.~Montanari, and S.~Oh.
\newblock Matrix completion from noisy entries.
\newblock In \emph{Advances in Neural Information Processing Systems}, pages
  952--960, 2009.

\bibitem[Keshavan et~al.(2010)Keshavan, Montanari, and Oh]{keshavan2010a}
R.~H. Keshavan, A.~Montanari, and S.~Oh.
\newblock Matrix completion from a few entries.
\newblock \emph{Information Theory, IEEE Transactions on}, 56\penalty0
  (6):\penalty0 2980--2998, 2010.

\bibitem[Koltchinskii et~al.(2011{\natexlab{a}})Koltchinskii, Lounici,
  Tsybakov, et~al.]{koltchinskii2011}
V.~Koltchinskii, K.~Lounici, A.~B. Tsybakov, et~al.
\newblock Nuclear-norm penalization and optimal rates for noisy low-rank matrix
  completion.
\newblock \emph{The Annals of Statistics}, 39\penalty0 (5):\penalty0
  2302--2329, 2011{\natexlab{a}}.

\bibitem[Koltchinskii et~al.(2011{\natexlab{b}})]{koltchinskii2011solo}
V.~Koltchinskii et~al.
\newblock Von neumann entropy penalization and low-rank matrix estimation.
\newblock \emph{The Annals of Statistics}, 39\penalty0 (6):\penalty0
  2936--2973, 2011{\natexlab{b}}.

\bibitem[Koren et~al.(2009)Koren, Bell, and Volinsky]{koren2009matrix}
Y.~Koren, R.~Bell, and C.~Volinsky.
\newblock Matrix factorization techniques for recommender systems.
\newblock \emph{Computer}, 42\penalty0 (8):\penalty0 30--37, 2009.

\bibitem[Li(1998{\natexlab{a}})]{li1998one}
R.-C. Li.
\newblock Relative perturbation theory: I. eigenvalue and singular value
  variations.
\newblock \emph{SIAM Journal on Matrix Analysis and Applications}, 19\penalty0
  (4):\penalty0 956--982, 1998{\natexlab{a}}.

\bibitem[Li(1998{\natexlab{b}})]{li1998two}
R.-C. Li.
\newblock Relative perturbation theory: II. eigenspace and singular subspace
  variations.
\newblock \emph{SIAM Journal on Matrix Analysis and Applications}, 20\penalty0
  (2):\penalty0 471--492, 1998{\natexlab{b}}.

\bibitem[Mazumder et~al.(2010)Mazumder, Hastie, and Tibshirani]{mazumder2010}
R.~Mazumder, T.~Hastie, and R.~Tibshirani.
\newblock Spectral regularization algorithms for learning large incomplete
  matrices.
\newblock \emph{The Journal of Machine Learning Research}, 11:\penalty0
  2287--2322, 2010.

\bibitem[Montanari and Oh(2010)]{montanari2010}
A.~Montanari and S.~Oh.
\newblock On positioning via distributed matrix completion.
\newblock In \emph{Sensor Array and Multichannel Signal Processing Workshop
  (SAM), 2010 IEEE}, pages 197--200. IEEE, 2010.

\bibitem[Negahban and Wainwright(2012)]{negahban2012}
S.~Negahban and M.~J. Wainwright.
\newblock Restricted strong convexity and weighted matrix completion: Optimal
  bounds with noise.
\newblock \emph{The Journal of Machine Learning Research}, 13\penalty0
  (1):\penalty0 1665--1697, 2012.

\bibitem[Negahban et~al.(2011)Negahban, Wainwright, et~al.]{negahban2011}
S.~Negahban, M.~J. Wainwright, et~al.
\newblock Estimation of (near) low-rank matrices with noise and
  high-dimensional scaling.
\newblock \emph{The Annals of Statistics}, 39\penalty0 (2):\penalty0
  1069--1097, 2011.

\bibitem[Recht(2011)]{recht2011}
B.~Recht.
\newblock A simpler approach to matrix completion.
\newblock \emph{The Journal of Machine Learning Research}, 12:\penalty0
  3413--3430, 2011.

\bibitem[Rennie and Srebro(2005)]{rennie2005}
J.~D. Rennie and N.~Srebro.
\newblock Fast maximum margin matrix factorization for collaborative
  prediction.
\newblock In \emph{Proceedings of the 22nd international conference on Machine
  learning}, pages 713--719. ACM, 2005.

\bibitem[Rohde et~al.(2011)Rohde, Tsybakov, et~al.]{rohde2011}
A.~Rohde, A.~B. Tsybakov, et~al.
\newblock Estimation of high-dimensional low-rank matrices.
\newblock \emph{The Annals of Statistics}, 39\penalty0 (2):\penalty0 887--930,
  2011.

\bibitem[Srebro et~al.(2004)Srebro, Rennie, and Jaakkola]{srebro2004}
N.~Srebro, J.~Rennie, and T.~S. Jaakkola.
\newblock Maximum-margin matrix factorization.
\newblock In \emph{Advances in neural information processing systems}, pages
  1329--1336, 2004.

\bibitem[Vershynin(2010)]{vershynin2010}
R.~Vershynin.
\newblock Introduction to the non-asymptotic analysis of random matrices.
\newblock \emph{arXiv preprint arXiv:1011.3027}, 2010.

\bibitem[Vu and Lei(2013)]{vu2013}
V.~Q. Vu and J.~Lei.
\newblock Minimax sparse principal subspace estimation in high dimensions.
\newblock \emph{The Annals of Statistics}, 41\penalty0 (6):\penalty0
  2905--2947, 2013.

\bibitem[Weinberger and Saul(2006)]{weinberger2006}
K.~Q. Weinberger and L.~K. Saul.
\newblock Unsupervised learning of image manifolds by semidefinite programming.
\newblock \emph{International Journal of Computer Vision}, 70\penalty0
  (1):\penalty0 77--90, 2006.



\end{thebibliography}

\newpage
\appendix
\numberwithin{equation}{section}
\section{Appendix}

\subsection{Proofs for Lemma \ref{lem0}} \label{apdx0}

\begin{proof}[Proof of Lemma \ref{lem0}]
Let 
$$\My = \left[ ( \y_{kh} - \p) {\M}_{kh} \right]_{1\le k\le\n, 1\le h\le\d} \;\text{ and }\; 
	\epsilony =\\ \left[ \y_{kh} \epsilon_{kh} \right]_{1\le k\le\n, 1\le h\le\d},$$
both in  $\real^{\n \times \d}$. 
Then, $\Mp = \p \M + \My + \epsilony$ and
\begin{eqnarray}\label{Sigp}
\hat{\Sig} 
&=& \p^2 \M^T\M + \My^T\My + \epsilony^T\epsilony \cr
&&\quad\quad
	+ \p \M^T\My + \p \My^T\M + \p \M^T\epsilony + \p \epsilony^T\M + \My^T\epsilony + \epsilony^T\My.
\end{eqnarray}
The result \eqref{expect-Sigp} follows since under the model setup in Section \ref{setup},
$$ \expect \My = 0, \;\; \expect \epsilony = 0, \;\; \expect (\My^T\My) = \p(1-\p)\,\diag(\M^T\M), \;\;
	\expect (\epsilony^T\epsilony) = \n\p\sig^2 I_\d, $$
$$ \expect (\M^T\My) = 0
,\quad
\expect (\M^T\epsilony) = 0
,\quad\text{and}\quad
\expect (\My^T\epsilony) = 0. $$

We can similarly show the result \eqref{expect-Sigpt}.
\end{proof}

\subsection{Proofs for Section \ref{proofs:thm1}} \label{apdx1}

\begin{proof}[Proof of Proposition \ref{prop1}]
We only show the result \eqref{prop1:res1}, since the other result can be shown similarly.

Let
$$E = \left\{ \frac{1}{\n\d} | {\ddot{\lam_\p}^2}_{\m+1} - {\lam_\p^2}_{\m+1} | < t \right\},$$
where $t = C_1 \p \frac{\log\n}{\d} + C_2 \p^{3/2}\sqrt{\frac{\log \n}{\n}}$. Note that $\frac{t}{\p^2} \to 0$.
By Weyl's theorem (\cite{li1998one}) and Lemma \ref{lem6}, we have for large constants $C_1,C_2>0$,
$$\prob(E^c) \le \prob\(\frac{1}{\n\d} \norm{\hat{\Sig}_\p - \expect \hat{\Sig}_\p}_2 \ge t\) = O(\n^{-2}).$$
Thus, for large $\n$,
\begin{eqnarray} \label{prop1:eq4}
&&\expect \norm{\sin\big( \V_\p^{(\m)},\V^{(\m)}\big)}_F^2
\cr
&&= \expect \bigg\{ \norm{\sin\big( \V_\p^{(\m)},\V^{(\m)}\big)}_F^2 \; \I_{E^c} \bigg\}
	+ \expect \bigg\{ \norm{\sin\big( \V_\p^{(\m)},\V^{(\m)}\big)}_F^2 \; \I_{E} \bigg\}
\cr
&&\le \m\, \prob \( E^c \)
	+ \expect \Bigg\{ \frac{ \norm{ \frac{1}{\n\d} \frac{1}{\p^2}\(\hat{\Sig}_\p - \expect \, \hat{\Sig}_\p\)\V^{(\m)} }_F^2}
		{ \left( \frac{1}{\n\d} \frac{1}{\p^2}| {\ddot{\lam_\p}^2}_{\m} - {\lam_\p^2}_{\m+1} | \right)^2} 
		\; \I_{E} \Bigg\}
\cr
&&\le \m\, \prob \( E^c \)
	+ \expect \Bigg\{ \frac{ \norm{ \frac{1}{\n\d} \frac{1}{\p^2}\(\hat{\Sig}_\p - \expect \, \hat{\Sig}_\p\)\V^{(\m)} }_F^2 \; \I_{E} }{ \left( \frac{1}{\n\d} \frac{1}{\p^2}
		\left| {\ddot{\lam_\p}^2}_{\m} - {\ddot{\lam_\p}^2}_{\m+1}\right| - \frac{t}{\p^2} \right)^2} \Bigg\}
\cr
&&\le \m\, \prob \( E^c \)
	+ \expect \Bigg\{ \frac{ \norm{ \frac{1}{\n\d} \frac{1}{\p^2}\(\hat{\Sig}_\p - \expect \, \hat{\Sig}_\p\)\V^{(\m)} }_F^2 \; \I_{E} }{ \left( \frac{1}{2\n\d} \frac{1}{\p^2}
		\left| {\ddot{\lam_\p}^2}_{\m} - {\ddot{\lam_\p}^2}_{\m+1} \right| \right)^2} \Bigg\}
\cr
&&\le c\n^{-2} + \frac{C\n^{-1}}{\p\,({b}_{\m}^2 - {b}_{\m+1}^2)^2},
\end{eqnarray}
where $\I_E$ is an indicator function of an event $E$, 
	the first inequality is due to the fact that $\|\sin (\Vhat^{(\m)}, \V_\p^{(\m)})\|_F^2 \le \m$ and 
	Davis-Kahan $\sin \theta$ theorem (Theorem 3.1 in \cite{li1998two}),
	and the last inequality holds by Lemma \ref{lem7} below.
\end{proof}

\begin{lemma} \label{lem7}
Under the model setup in Section \ref{setup} and Assumption \ref{assume1}, we have for large $\n$ and $\d$,
\begin{equation} \label{lem7:res1} 
	\expect \norm{ \frac{1}{\n\d} \; \frac{1}{\p^2} \( \hat{\Sig}_\p-\expect\hat{\Sig}_\p \) \V^{(\m)} }_F^2 \le \frac{C_1}{\p\,\n}
\end{equation}
and
\begin{equation*}
	\expect \norm{ \frac{1}{\n\d} \; \frac{1}{\p^2} \( \hat{\Sig}_{\p t}-\expect \hat{\Sig}_{\p t} \) \U^{(\m)} }_F^2 \le \frac{C_2}{\p\,\d},
\end{equation*}
where $\hat{\Sig}_{\p}$ and $\hat{\Sig}_{\p t}$ are defined in \eqref{Sigpp} and $C_1$ and $C_2$ are generic constants free of $n, d,$ and $p$.
\end{lemma}
\begin{proof}[Proof of Lemma \ref{lem7}]
We only show the result \eqref{lem7:res1} because the other result holds similarly. 

From \eqref{Sigp}, \eqref{Sigpp}, Proposition \ref{prop0}, and triangle inequality, we have 
\begin{eqnarray}\label{lem7:eq1}
&&\norm{ \(\hat{\Sig}_\p - \expect\hat{\Sig}_\p\)\V^{(\m)} }_F
\cr
&&\le\norm{ \big[ \My^T\My -(1-\p) \diag(\My^T\My) - \p^2(1-\p)\,\diag(\M^T\M) \big]\V^{(\m)} }_F
\cr
&&+ \norm{ \big[ \epsilony^T\epsilony -(1-\p) \diag(\epsilony^T\epsilony)- \n\p^2\sig^2 I_\d \big]\V^{(\m)} }_F
\cr
&&+ \p\norm{ \big[ \My^T\M -(1-\p) \diag(\My^T\M) \big]\V^{(\m)} }_F
\cr
&&+ \p\norm{ \big[ \M^T\My -(1-\p) \diag(\M^T\My) \big]\V^{(\m)} }_F
\cr
&&+ \p\norm{ \big[ \epsilony^T\M -(1-\p) \diag(\epsilony^T\M) \big]\V^{(\m)} }_F
\cr
&&+ \p\norm{ \big[ \M^T\epsilony -(1-\p) \diag(\M^T\epsilony) \big]\V^{(\m)} }_F
\cr
&&+ \norm{ \big[ \My^T\epsilony -(1-\p) \diag(\My^T\epsilony) \big]\V^{(\m)} }_F
\cr
&&+ \norm{ \big[ \epsilony^T\My -(1-\p) \diag(\epsilony^T\My) \big]\V^{(\m)} }_F
\cr
&&= (A) + (B) + \p\:(C) + \p\:(D) + \p\:(E) + \p\:(F) + (G) + (H) .
\end{eqnarray}
We examine the convergence rates of the above terms, $(A)$-$(H)$. 

First, consider the term (A) in \eqref{lem7:eq1}. Then, we have
\begin{eqnarray} \label{lem7:eq1(A)}
&&\expect \norm{ \big[ \My^T\My -(1-\p) \diag(\My^T\My) - \p^2(1-\p)\,\diag(\M^T\M) \big]\V^{(\m)} }_F^2
\cr
&&= \sum_{i=1}^{\d} \sum_{j=1}^{\m} \expect \bigg\{ \sum_{k=1}^{\n}\sum_{h=1}^{\d} \Big[ 
	\p\, \Big((\y_{ki}-\p)^2 - \p(1-\p)\Big) {\M}_{ki}^2 \V_{ji} \; \I_{(h=i)}
\cr
&&\quad\quad\quad\quad\quad\quad\quad\quad\quad\quad
	+ (\y_{ki}-\p)(\y_{kh}-\p) {\M}_{ki}{\M}_{kh} \V_{jh}\; \I_{(h\ne i)} \Big]\bigg\}^2
\cr
&&= \sum_{i=1}^{\d} \sum_{j=1}^{\m} \Bigg\{ \sum_{k=1}^{\n}\sum_{h=1}^{\d} \bigg[ 
	\p^2\, \expect \Big((\y_{ki}-\p)^2 - \p(1-\p)\Big)^2 {\M}_{ki}^4 \V_{ji}^2 \; \I_{(h=i)}
\cr
&&\quad\quad\quad\quad\quad\quad\quad\quad\quad
	+ \expect \Big((\y_{ki}-\p)^2(\y_{kh}-\p)^2\Big) {\M}_{ki}^2{\M}_{kh}^2 \V_{jh}^2\; \I_{(h\ne i)} \bigg]\Bigg\}
\cr
&&= \sum_{i=1}^{\d} \sum_{j=1}^{\m} \Bigg\{ \sum_{k=1}^{\n}\sum_{h=1}^{\d} \bigg[ 
	\p^3(1-\p)(2\p-1)^2 {\M}_{ki}^4 \V_{ji}^2 \; \I_{(h=i)}
\cr
&&\quad\quad\quad\quad\quad\quad\quad\quad\quad\quad\quad\quad
	+ \p^2(1-\p)^2 {\M}_{ki}^2{\M}_{kh}^2 \V_{jh}^2 \; \I_{(h\ne i)} \bigg]\Bigg\}
\cr
&&\le \p^2(1-\p) \L^4 \sum_{i=1}^{\d} \sum_{j=1}^{\m} \sum_{k=1}^{\n}\sum_{h=1}^{\d} \V_{jh}^2
\cr
&&= \p^2(1-\p) \L^4 \sum_{i=1}^{\d} \sum_{j=1}^{\m} \sum_{k=1}^{\n} 1
\cr
&&\le C \p^2(1-\p) \n\d.
\end{eqnarray}
Similarly to \eqref{lem7:eq1(A)}, we can show that the expected values of the terms $(B), (D), (F), (G)$, and $(H)$ squared are bounded by $C\p^2\n\d$.

Second, consider the term (C) in \eqref{lem7:eq1}. Then, we have
\begin{eqnarray} \label{lem7:eq1(C)}
&&\expect \norm{ \big[ \My^T\M -(1-\p) \diag(\My^T\M) \big]\V^{(\m)} }_F^2
\cr
&&= \sum_{i=1}^{\d} \sum_{j=1}^{\m} \expect \Bigg\{ \sum_{k=1}^{\n} (\y_{ki}-\p) \, 
	\sum_{h=1}^{\d} \bigg[ \p\, {\M}_{ki}^2 \V_{ji} \, \I_{(h=i)}
\cr
&&\hspace*{6.5cm}
	+ {\M}_{ki}{\M}_{kh} \V_{jh}\, \I_{(h\ne i)} \bigg]\Bigg\}^2
\cr
&&= \sum_{i=1}^{\d} \sum_{j=1}^{\m} \sum_{k=1}^{\n} \expect(\y_{ki}-\p)^2 \; 
	\Bigg\{ \sum_{h=1}^{\d} {\M}_{ki}{\M}_{kh} \V_{jh} \Big[ 1 - (1-\p) \I_{(h=i)} \Big]\Bigg\}^2
\cr
&&= \p(1-\p)\sum_{i=1}^{\d} \sum_{j=1}^{\m} \sum_{k=1}^{\n}
	\Bigg\{ \sum_{h=1}^{\d} {\M}_{ki}{\M}_{kh} \V_{jh} \Big[ 1 - (1-\p) \I_{(h=i)} \Big]\Bigg\}^2
\cr
&&\le \p(1-\p)\L^4 \sum_{i=1}^{\d} \sum_{j=1}^{\m} \sum_{k=1}^{\n} \Bigg\{ \sum_{h=1}^{\d} 
	|\V_{jh}| \Bigg\}^2
\cr
&&\le C\p(1-\p)\n\d^2,
\end{eqnarray}
where the last inequality holds due to Cauchy-Schwarz inequality.

Lastly, for the term (E) in \eqref{lem7:eq1}, 
\begin{eqnarray} \label{lem7:eq1(E)}
&&\expect \norm{ \big[ \epsilony^T\M -(1-\p) \diag(\epsilony^T\M) \big]\V^{(\m)} }_F^2
\cr
&&= \sum_{i=1}^{\d} \sum_{j=1}^{\m} \expect \Bigg\{ \sum_{k=1}^{\n} \y_{ki}\epsilon_{ki} \; 
	\sum_{h=1}^{\d} {\M}_{kh} \V_{jh} \Big[ 1 - (1-\p) \I_{(h=i)} \Big] \Bigg\}^2
\cr
&&= \sum_{i=1}^{\d} \sum_{j=1}^{\m} \sum_{k=1}^{\n} \expect \big( \y_{ki}^2\epsilon_{ki}^2\big) \; 
	\Bigg\{ \sum_{h=1}^{\d} {\M}_{kh} \V_{jh} \Big[ 1 - (1-\p) \I_{(h=i)} \Big] \Bigg\}^2
\cr
&&= \p\sig^2\sum_{i=1}^{\d} \sum_{j=1}^{\m} \sum_{k=1}^{\n}
	\Bigg\{ \sum_{h=1}^{\d} {\M}_{kh} \V_{jh} \Big[ 1 - (1-\p) \I_{(h=i)} \Big] \Bigg\}^2
\cr
&&\le \p\sig^2\L^2 \sum_{i=1}^{\d} \sum_{j=1}^{\m} \sum_{k=1}^{\n}
	\Bigg\{ \sum_{h=1}^{\d} |\V_{jh}| \Bigg\}^2
\cr
&&\le C\p\n\d^2,
\end{eqnarray}
where last inequality holds due to Cauchy-Schwarz inequality.

The result follows from \eqref{lem7:eq1(A)}-\eqref{lem7:eq1(E)}.
\end{proof}

\begin{lemma} \label{lem9}
Under the model setup in Section \ref{setup} and Assumption \ref{assume1}, we have for any given $\xi_1>0$,
\begin{equation*}
\norm{\My}_2 \le C_{\xi_1} \sqrt{\p\, \n \log \n}
\end{equation*}
with probability $1-O(\n^{-\xi_1})$. Similarly, we have for any given $\xi_2>0$,
\begin{equation*}
\norm{\epsilony}_2 \le C_{\xi_2} \sqrt{\p\, \n \log \n}
\end{equation*}
with probability $1-O(\n^{-\xi_2})$.
\end{lemma}

\begin{proof}[Proof of Lemma \ref{lem9}]
Let $\My^{(i,j)} \in \real^{\n\times\d}$ be such that 
$${\My}_{kh}^{(i,j)}=
\left\{\begin{matrix}
(\y_{kh}-\p){\M}_{kh}, & (k,h)=(i,j)\\ 
0, & (k,h)\ne(i,j)
\end{matrix}\right.
\;\text{ for } 
1\le k\le \n \text{ and } 1\le h\le \d.
$$
Then, $$\frac{1}{\n\d}\My = \frac{1}{\n\d}\sum_{i=1}^\n\sum_{j=1}^\d\My^{(i,j)},$$ 
$\expect(\My^{(i,j)})=0$, and $\smallnorm{\My^{(i,j)}}_2\le\L$ for all $1\le k\le \n$ and $1\le h\le \d$.
Also, we have 
\begin{eqnarray}
&&\norm{ \frac{1}{\n\d}\expect\(\My^{(i,j)}\My^{(i,j)T}\) }_2 = \norm{ \frac{\p(1-\p)}{\n\d} \diag\(\M\M^T\) }_2 \le \frac{\p\,\L^2}{\n} \text{ and}
\cr
&&\norm{ \frac{1}{\n\d}\expect\(\My^{(i,j)T}\My^{(i,j)}\) }_2 = \norm{ \frac{\p(1-\p)}{\n\d} \diag\(\M^T\M\) }_2 \le \frac{\p\,\L^2}{\d}.
\end{eqnarray}
Thus, by Proposition 1 in \cite{koltchinskii2011}, we have 
$$\norm{\frac{1}{\n\d}\My}_2 \le C \max \( \sqrt{\frac{\p\,\L^2}{\d}}\sqrt{\frac{\log \n}{\n\d}},\L\frac{\log\n}{\n\d}\) \le C\sqrt{\frac{\p\,\log\n}{\n\d^2}}
$$
with probability at least $1 - \n^{-\xi_1}$.

In a similar way together with Proposition 2 in \cite{koltchinskii2011}, we can show that 
$\norm{\frac{1}{\n\d}\epsilony}_2\le C\sqrt{\frac{\p\,\log\n}{\n\d^2}}$
with probability at least $1 - \n^{-\xi_2}$.

\end{proof}

\begin{proof}[Proof of Lemma \ref{lem6}]
We only show the result \eqref{lem6:res1} because the other result holds similarly.

From \eqref{Sigp}, Proposition \ref{prop0} and triangle inequality, we have
\begin{eqnarray} \label{lem6:eq1}
&&\frac{1}{\n\d}\norm{ \hat{\Sig}_\p-\expect\hat{\Sig}_\p } _2
\cr
&&\le\frac{1}{\n\d} \norm{ \My^T\My -(1-\p) \diag(\My^T\My) - \p^2(1-\p)\,\diag(\M^T\M) } _2
\cr
&&+ \frac{1}{\n\d} \norm{ \epsilony^T\epsilony -(1-\p) \diag(\epsilony^T\epsilony)- \n\p^2\sig^2 I_\d } _2 
\cr
&&+ 2\,\frac{1}{\n\d} \norm{ \p \My^T\M -(1-\p) \p \diag(\My^T\M) } _2 
\cr
&&+ 2\,\frac{1}{\n\d} \norm{ \p \epsilony^T\M -(1-\p) \p \diag(\epsilony^T\M) } _2 
\cr
&&+ 2\,\frac{1}{\n\d} \norm{ \My^T\epsilony -(1-\p) \diag(\My^T\epsilony) } _2 
\cr
&&= (I) + (II) + 2\;(III) + 2\;(IV) + 2\;(V).
\end{eqnarray}
Because of similarity, we provide arguments only for $(I)$ and $(IV)$.

Consider the term $(I)$ in \eqref{lem6:eq1}. 
First, we have by Lemma \ref{lem9}
\begin{eqnarray} \label{lem6:MyMy}
\frac{1}{\n\d} \norm{ \My^T\My }_2
= \n\d\norm{\frac{1}{\n\d}\My}_2^2
\le C \p \frac{\log\n}{\d}
\end{eqnarray}
with probability at least $1-O(\n^{-\mu_1})$.
Also, we have with probability at least $1-O(\n^{-\mu_1})$,
\begin{eqnarray} \label{lem6:MyMy-diag}
&&\frac{1-\p}{\n\d}\norm{ \diag(\My^T\My) + \p^2\,\diag(\M^T\M) }_2
\cr
&&\le \frac{1-\p}{\n\d}\norm{ \diag(\My^T\My) - \p(1-\p)\,\diag(\M^T\M) }_2
\cr
&&\quad\quad
	+ \frac{\p(1-\p)}{\n\d}\norm{ \diag(\M^T\M) }_2
\cr
&&= (1-\p) \max_{1\le h\le \d} \left| \sum_{k=1}^\n \frac{ \[ (\y_{kh}-\p)^2-\p(1-\p) \]{\M}_{kh}^2 }{\n\d} \right|
\cr 
&&\quad\quad
	+ \frac{\p(1-\p)}{\n\d} \max_{1\le h\le \d} \left| \sum_{k=1}^\n {\M}_{kh}^2 \right|
\cr
&&\le C \sqrt{\frac{\p\log \n}{\n}}\frac{1}{\d} + \frac{\p(1-\p)\L^2}{\d}
\cr
&&\le C \p\d^{-1},
\end{eqnarray}
where the second inequality holds by \eqref{lem6:MyMy-diag-supp} below. 
Take $t^2=c\frac{\log\n}{\n\d^2}\p(1-\p)(3\p^2-3\p+1)$ for some large constant $c>0$. 
Then, by Bernstein's inequality,
\begin{eqnarray}\label{lem6:MyMy-diag-supp}
&& \prob\( \max_{1\le h\le \d}\left| \sum_{k=1}^\n \frac{ \[ (\y_{kh}-\p)^2-\p(1-\p) \]{\M}_{kh}^2 }{\n\d} \right| \ge t\)
\cr
&&\le \sum_{h=1}^{\d} \prob\( \left| \sum_{k=1}^\n \[ (\y_{kh}-\p)^2-\p(1-\p) \]{\M}_{kh}^2 \right| \ge \n\d t \)
\cr
&&\le 2 \d \; \exp \left\{ - \frac{\n\d^2 t^2}{2\L^4 \p(1-\p)(3\p^2-3\p+1)} \right\}
\cr
&&= C \n^{-\mu_1}.
\end{eqnarray}
By \eqref{lem6:MyMy} and \eqref{lem6:MyMy-diag}, we have 
\begin{eqnarray} \label{lem6:result(I)}
(I)\le C \p \frac{\log\n}{\d}
\end{eqnarray} 
with probability at least $1-O(\n^{-\mu_1})$.
Similarly, we can show that $(II)$ and $(V)$ are bounded by $C \p \frac{\log\n}{\d}$ with probability at least $1-O(\n^{-\mu_1})$.

Consider the term $(IV)$ in \eqref{lem6:eq1}.
We have
\begin{eqnarray*}
&& (IV)^2 \le \Bigg\{ \max_{1\le j\le \d}  \sum_{i=1}^{\d} \left| \sum_{k=1}^\n X_{kij}^{(IV)} \right| \Bigg\}
			\Bigg\{ \max_{1\le i\le \d}  \sum_{j=1}^{\d} \left| \sum_{k=1}^\n X_{kij}^{(IV)} \right| \Bigg\},
\end{eqnarray*}
where $\n\d\, X_{kij}^{(IV)} = \p\, \y_{ki}{\epsilon}_{ki}{\M}_{kj}\I_{(i\ne j)} 
	+ \p^2 \y_{ki}{\epsilon}_{ki}{\M}_{kj}\I_{(i=j)}$ and hence $X_{kij}^{(IV)}$ 
are centered sub-Gaussian random variables under the model setup in Section \ref{setup}.
Then, we have for any $\rho \in \real$ and for all $1\le k\le\n$, $1\le i\le d,$ and $1\le j\le\d$,
$$\expect \exp \left\{ \rho X_{kij}^{(IV)} \right\} \le \exp \left\{ \frac{\rho^2\frac{\p^3\beta}{\n^2\d^2}}{2} \right\} 
	\;\text{ for some constant } \beta>0.$$
Take $t^2 = c\p^3\frac{\log\n}{\n}$ for some large constant $c>0$ and $\rho = \frac{t/\d}{\n\frac{\p^3\beta}{\n^2\d^2}}$. 
Then, by Markov's inequality,
\begin{eqnarray} \label{prop4:eq3}
\prob\( \max_{1\le j\le\d} \sum_{i=1}^\d \left| \sum_{k=1}^\n X_{kij}^{(IV)} \right| > t \)
&\le& \sum_{j=1}^\d \sum_{i=1}^\d \prob \( \left| \sum_{k=1}^\n X_{kij}^{(IV)} \right| > t/\d \) \cr
&\le& 2\sum_{j=1}^{\d} \sum_{i=1}^{\d} \frac{\expect \( \exp\left\{\rho \sum_{k=1}^\n X_{kij}^{(IV)}\right\}\)}{\exp\left\{\rho(t/\d)\right\}} \cr
&\le& 2 \d^2 \; \exp \left\{ -\rho\frac{t}{\d} + \frac{\rho^2}{2} \frac{\p^3\beta}{\n\d^2} \right\}\cr
&=& 2\d^2 \; \exp \left\{ -\frac{\n t^2}{2\p^3\beta} \right\} \cr
&=& C\n^{-\mu_1}.
\end{eqnarray}
Similarly, 
\begin{eqnarray} \label{prop4:eq3-mirror}
\prob\( \max_{1\le i\le\d} \sum_{j=1}^\d \left| \sum_{k=1}^\n X_{kij}^{(IV)} \right| > t \)
	\le C\n^{-\mu_1}.
\end{eqnarray}
By \eqref{prop4:eq3} and \eqref{prop4:eq3-mirror}, with probability at least $1-O\(\n^{-\mu_1}\)$, 
\begin{eqnarray} \label{lem6:result(IV)}
\left| (IV) \right| \le C\p^{3/2}\sqrt{\frac{\log \n}{\n}}.
\end{eqnarray} 
Similarly, we can show that $(III)$ is bounded by $C \p^{3/2} \sqrt{\frac{\log\n}{\n}}$ with probability at least $1-O(\n^{-\mu_1})$.

The statement is showed by \eqref{lem6:result(I)} and \eqref{lem6:result(IV)}.
\end{proof}

\begin{proof}[Proof of Lemma \ref{lem4}]
We only show the result \eqref{lem4:res1} because the other result holds similarly.

By triangle inequality, we have
\begin{eqnarray} \label{lem4:eq0}
&&\frac{1}{\n\d} \norm{\hat{\Sig}_{\phat} - \hat{\Sig}_\p}_2 
= \frac{1}{\n\d} \norm{\(\phat - \p\)\diag(\hat{\Sig})}_2
\cr 
&&\le \frac{|\phat - \p|}{\n\d} \Big\{ \smallnorm{\diag(\hat{\Sig}) - \diag(\p\M^T\M + \n\p\sig^2 I_\d)}_2 
\cr
&&\quad\quad\quad\quad\quad\quad\quad\quad\quad\quad
	+ \smallnorm{ \diag(\p\M^T\M + \n\p\sig^2 I_\d)}_2 \Big\}.
\end{eqnarray}
We will look at the terms in \eqref{lem4:eq0} one by one.

By Bernstein's inequality, we have for large constant $C>0$,
\begin{eqnarray} \label{lem4:eq1}
&&\prob\( \left|\phat - \p\right| \ge C \sqrt{\frac{\p(1-\p)\log \n}{\n\d}} \) 
\cr
&&= \prob\( \left|\sum_{k=1}^{\n}\sum_{h=1}^{\d} (\y_{kh}-\p)\right| \ge C \sqrt{\p(1-\p) \n\d \log \n} \)
\cr
&&\le 2 \, \exp\left\{-\nu_1 \log\n\right\}
\cr
&&= 2\n^{-\nu_1}.
\end{eqnarray}

Take $t^2 = c \frac{\p\log\n}{\n\d^2}$ for some large constant $c>0$. 
Then, since $\y_{ki}^2({\M}_{ki} + \epsilon_{ki})^2 - \p ({\M}_{ki}^2 + \sig^2), k=1,\ldots,\n$, 
	are independent centered sub-exponential random variables, 
we have by Proposition 5.16 in \cite{vershynin2010}, 
\begin{eqnarray} \label{lem4:eq2}
&&\prob\( \frac{1}{\n\d} \norm{ \diag(\hat{\Sig}) - \diag(\p \M^T\M + \n\p\sig^2 I_\d) }_2 \ge t\) 
\cr
&&= \prob\( \frac{1}{\n\d} \max_{1\le i\le \d}\left|\sum_{k=1}^{\n}\bigg[ \y_{ki}^2({\M}_{ki} + \epsilon_{ki})^2 - \p ({\M}_{ki}^2 + \sig^2) \bigg]\right| \ge t\)
\cr
&&\le \sum_{i=1}^{\d}\prob\( \left|\sum_{k=1}^{\n}\bigg[ \y_{ki}^2({\M}_{ki} + \epsilon_{ki})^2 - \p ({\M}_{ki}^2 + \sig^2) \bigg]\right| \ge \n\d t \)
\cr
&&\le 2 \d \; \exp \left\{ - \frac{\n^2\d^2 t^2}{c_1\n\p} \right\}
\cr
&&\le C \n^{-\nu_1}.
\end{eqnarray}

Also, note that
\begin{eqnarray} \label{lem4:eq3}
\norm{\frac{1}{\n\d} \diag(\p\M^T\M + \n\p\sig^2 I_\d) }_2
&=& \frac{1}{\n\d} \; \max_{1\le i\le \d} \; \p \sum_{k=1}^\n {\M}_{ki}^2 + \n\p\sig^2
\cr
&\le& \frac{\p( \L^2 + \sig^2 )}{\d} .
\end{eqnarray}

Combining the results in \eqref{lem4:eq0}-\eqref{lem4:eq3}, we have
\begin{eqnarray} \label{lem4:eq4}
\frac{1}{\n\d} \norm{\hat{\Sig}_{\phat} - \hat{\Sig}_\p}_2
	\le C\p^{3/2}\,\sqrt{\frac{\log\n}{\n\d}}\;\frac{1}{\d}
\end{eqnarray}
with probability at least $1-O\(\n^{-\nu_1}\)$. 
\end{proof}

\begin{proof}[Proof of Lemma \ref{lem5}]
We only show the result \eqref{lem5:res1} because the other result holds similarly. 

We have
\begin{eqnarray} \label{lem5:eq1}
&&\expect \norm{ \frac{1}{\n\d} \( \hat{\Sig}_{\phat}-\hat{\Sig}_\p \) \V^{(\m)} }_F^2
\cr
&&\le \m \; \expect \norm{ \frac{1}{\n\d} \( \hat{\Sig}_{\phat}-\hat{\Sig}_\p \) }_2^2
\cr
&&\le \m \; \expect \left\{ \( \phat - \p \)^2 \norm{\frac{1}{\n\d} \diag(\hat{\Sig})}_2^2 \right\}
\cr
&&\le 4\m \; \expect \left\{ \( \phat - \p \)^2 \norm{\frac{1}{\n\d} \diag(\hat{\Sig}) - \frac{1}{\n\d} \diag(\p\M^T\M + \n\p\sig^2 I_\d) }_2^2 \right\}
\cr
&&\quad\quad
	+ 4\m \norm{\frac{1}{\n\d} \diag(\p\M^T\M + \n\p\sig^2 I_\d) }_2^2 \; \expect \( \phat - \p \)^2  
\cr
&&\le 4\m \; \sqrt{ \expect \( \phat - \p \)^4 \; \expect \norm{\frac{1}{\n\d} \diag(\hat{\Sig}) - \frac{1}{\n\d} \diag(\p\M^T\M + \n\p\sig^2 I_\d)}_2^4 }
\cr
&&\quad\quad
	+ 4\m \; \frac{\p^2 (\L^2 + \sig^2)^2}{\d^2} \frac{\p(1-\p)}{\n\d}
\cr
&&\le C_1 \frac{\p^2(1-\p)}{\n^2\d^{5/2}} + C_2 \frac{\p^3(1-\p)}{\n\d^3},
\end{eqnarray}
where the fourth inequality holds by H\"older's inequality and the fifth inequality is due to the fact that 
\begin{eqnarray} \label{lem5:eq1}
&&\expect \( \phat - \p \)^4 \; \expect \norm{\frac{1}{\n\d} \diag(\hat{\Sig}) - \frac{1}{\n\d} \diag(\p\M^T\M + \n\p\sig^2 I_\d)}_2^4 
\cr
&&\le \frac{\expect \big\{\sum_{k=1}^\n\sum_{h=1}^\d \(\y_{kh} - \p\)\big\}^4}{\n^4d^4}
\cr
&&\quad\quad\quad
	\frac{\expect \left\{ \max_{1\le i\le \d} \left| \sum_{k=1}^\n 
		\left[ \y_{ki}^2 ({\M}_{ki}+\epsilon_{ki})^2 - \p({\M}_{ki}^2+\sig^2) \right] \right|^4 \right\} }{\n^4\d^4}
\cr
&&\le \frac{\expect \big\{\sum_{k=1}^\n\sum_{h=1}^\d \(\y_{kh} - \p\)\big\}^4  }{\n^4\d^4}
\cr
&&\quad\quad\quad
		\frac{\d\,  \expect \left[ \sum_{k=1}^\n \(\y_{ki}^2 ({\M}_{ki}+\epsilon_{ki})^2 - \p({\M}_{ki}^2+\sig^2)\) \right]^4 }{\n^4\d^4}
\cr
&&= \frac{O\(\p^2(1-\p)^2\n^2\d^2\) \, O\( \p^2\n^2 \d\) }{\n^8\d^8}.
\end{eqnarray}
\end{proof}

\subsection{Proofs for Section \ref{proofs:thm2}} \label{apdx2}

\begin{lemma} \label{prop2}
Under the model setup in Section \ref{setup} and Assumption \ref{assume1}, we have
\begin{eqnarray*} 
&&\sum_{i=1}^{\m} {\lam_\p^2}_i - \p^2 \left[ \sum_{i=1}^{\m}\lam^2_i + \n\sig^2 \right]
\cr
&&= 2\p\,\sum_{k=1}^\n \sum_{h=1}^\d (\y_{kh}-\p){\M}_{kh}  
	\sum_{i=1}^\m \V_{ih} \Bigg[ \bigg(\sum_{h'=1}^\d {\M}_{kh'} \V_{ih'}\bigg) - (1-\p){\M}_{kh} \V_{ih}\Bigg] \cr
&&\quad+ 2\p\,\sum_{k=1}^\n \sum_{h=1}^\d \y_{kh}\epsilon_{kh}  
	\sum_{i=1}^\m \V_{ih} \Bigg[ \bigg(\sum_{h'=1}^\d {\M}_{kh'} \V_{ih'}\bigg) - (1-\p){\M}_{kh} \V_{ih}\Bigg] \cr
&&\quad+o_p\(\sqrt{\n\d}\)
\cr
&&= (i) + (ii) +O_p\(\p\sqrt{\n}+\p\d\)
\end{eqnarray*} 
and $(i)+(ii)=O_p\(\sqrt{\p^3\n\d}\)$, where $\lam_{\p i}$ and $\lam_{i}$ are defined in \eqref{Sigpp} and \eqref{mod}, respectively.
\end{lemma}

\begin{proof}[Proof of Lemma \ref{prop2}]
We have 
\begin{eqnarray} \label{prop2:eq1}
&&\sum_{i=1}^{\m} {\lam_\p^2}_i - \p^2 \left[ \sum_{i=1}^{\m}\lam^2_i + \n\sig^2 \right] 
\cr
&&=\trace({\V_\p}^{(\m)T} \hat{\Sig}_\p {\V_\p}^{(\m)}) - \trace\(\V^{(\m)T} (\p^2 \M^T\M + \n\p^2\sig^2 I_\d) \V^{(\m)} \)
\cr 
&&=\trace(\O^T\V^{(\m)T} \hat{\Sig}_\p \V^{(\m)}\O ) \cr
&&\quad\quad+ \trace({\V_\p}^{(\m)T} \hat{\Sig}_\p {\V_\p}^{(\m)} - \O^T\V^{(\m)T} \hat{\Sig}_\p \V^{(\m)}\O ) 
\cr
&&\quad\quad -\trace\(\V^{(\m)T} (\p^2 \M^T\M + \n\p^2\sig^2 I_\d) \V^{(\m)} \)
\cr
&&=\trace(\V^{(\m)T} \hat{\Sig}_\p \V^{(\m)} ) + \trace({\V_\p}^{(\m)T} \hat{\Sig}_\p {\V_\p}^{(\m)} - \O^T\V^{(\m)T} \hat{\Sig}_\p \V^{(\m)}\O ) 
\cr
&&\quad\quad -\trace\(\V^{(\m)T} (\p^2 \M^T\M + \n\p^2\sig^2 I_\d) \V^{(\m)} \)
\cr
&&=\trace\Big(\V^{(\m)T} (\My^T\My -(1-\p) \diag(\My^T\My) \cr
&&\quad\quad\quad\quad\quad\quad\quad\quad- \p^2(1-\p)\,\diag(\M^T\M)) \V^{(\m)} \Big) \cr
&&\quad\quad+\trace\(\V^{(\m)T} (\epsilony^T\epsilony -(1-\p) \diag(\epsilony^T\epsilony)- \n\p^2\sig^2 I_\d) \V^{(\m)} \) \cr
&&\quad\quad+\trace\Big(\V^{(\m)T} (\p \M^T\My + \p \My^T\M \cr
&&\quad\quad\quad\quad\quad\quad\quad\quad-(1-\p) \p \diag(\M^T\My + \My^T\M)) \V^{(\m)} \Big) \cr
&&\quad\quad+\trace\Big(\V^{(\m)T} (\p \M^T\epsilony + \p \epsilony^T \M \cr
&&\quad\quad\quad\quad\quad\quad\quad\quad-(1-\p) \p \diag(\M^T\epsilony + \epsilony^T\M)) \V^{(\m)} \Big) \cr
&&\quad\quad+\trace\Big(\V^{(\m)T} (\My^T\epsilony + \epsilony^T\My \cr
&&\quad\quad\quad\quad\quad\quad\quad\quad-(1-\p) \diag(\My^T\epsilony + \epsilony^T\My)) \V^{(\m)} \Big) \cr
&&\quad+ \trace({\V_\p}^{(\m)T} \hat{\Sig}_\p {\V_\p}^{(\m)} - \O^T\V^{(\m)T} \hat{\Sig}_\p \V^{(\m)}\O ) \cr
&&=(a)+(b)+(c)+(d)+(e)+(f),
\end{eqnarray}
where $\O \in \mathbb{V}_{\m,\m}$ is a solution to 
	$\inf_{\Q \in \mathbb{V}_{\m,\m}} \smallnorm{\V_\p^{(\m)}-\V^{(\m)}\Q}_F^2$
	and the fourth equality holds by \eqref{Sigpp} and \eqref{Sigp}. 
Below, we examine the six terms $(a)$-$(f)$ one by one.

The term $(a)$ in \eqref{prop2:eq1} is 
\begin{eqnarray} \label{prop2:eq2(a)-1}
(a)
&=& \sum_{i=1}^\m \V_i^T \(\My^T\My -(1-\p) \diag(\My^T\My) - \p^2(1-\p)\,\diag(\M^T\M)\) \V_i
\cr
&=&\sum_{i=1}^\m \Bigg\{ 
	\sum_{k=1}^\n \( \sum_{h=1}^\d (\y_{kh} - \p) {\M}_{kh} \V_{ih} \)^2
\cr
&&\quad\quad\quad\quad\quad
		 - (1-\p) \sum_{k=1}^\n \sum_{h=1}^\d (\y_{kh} - \p)^2 {\M}_{kh}^2 \V_{ih}^2
\cr
&&\quad\quad\quad\quad\quad
	- \p^2 (1-\p) \sum_{k=1}^\n \sum_{h=1}^\d {\M}_{kh}^2 \V_{ih}^2 \Bigg\}
\cr
&=&\sum_{k=1}^\n \sum_{h=1}^\d \p \Big[ (\y_{kh} - \p)^2 - \p(1-\p) \Big] {\M}_{kh}^2 \sum_{i=1}^\m \V_{ih}^2
\cr
&&\quad\quad\quad\quad\quad
	+ 2 \sum_{k=1}^\n \sum_{h<h'}^{1\sim\d} (\y_{kh} - \p)(\y_{kh'} - \p) {\M}_{kh}{\M}_{kh'} \sum_{i=1}^\m \V_{ih}\V_{ih'} .
\end{eqnarray}
Note that the two terms in \eqref{prop2:eq2(a)-1} are centered and uncorrelated with each other. 
So, the variance is 
\begin{eqnarray} \label{prop2:eq2(a)}
\var(a)
&=&\Bigg\{ \sum_{k=1}^\n \sum_{h=1}^\d \p^3(1-\p)(2\p-1)^2 {\M}_{kh}^4 \Bigg( \sum_{i=1}^\m \V_{ih}^2 \Bigg)^2 \Bigg\}
\cr
&&\quad\quad\quad\quad\quad
	+ \Bigg\{ 4 \sum_{k=1}^\n \sum_{h<h'}^{1\sim\d} \p^2(1-\p)^2 {\M}_{kh}^2{\M}_{kh'}^2 \Bigg(\sum_{i=1}^\m \V_{ih}\V_{ih'}\Bigg)^2 \Bigg\} 
\cr
&\le&\m \sum_{i=1}^\m \sum_{k=1}^\n \sum_{h=1}^\d \p^3(1-\p)(2\p-1)^2 {\M}_{kh}^4 \V_{ih}^4 
\cr
&&+ 4\m \sum_{i=1}^\m \sum_{k=1}^\n \sum_{h<h'}^{1\sim\d} \p^2(1-\p)^2 {\M}_{kh}^2{\M}_{kh'}^2 \V_{ih}^2\V_{ih'}^2
\cr
&\le& \m\L^4\p^3(1-\p)(2\p-1)^2 \sum_{i=1}^\m \sum_{k=1}^\n \sum_{h=1}^\d \V_{ih}^4 
\cr
&&+ 4\m\L^4\p^2(1-\p)^2 \sum_{i=1}^\m \sum_{k=1}^\n \sum_{h,h'}^{1\sim\d} \V_{ih}^2\V_{ih'}^2
\cr
&\le& C\p^2(1-\p)\n,
\end{eqnarray}
where the first inequality is due to Jensen's inequality. 
This shows that the term $(a)$ is $O_p(\p\sqrt{\n})$. 
Similarly, we can show that the terms $(b)$ and $(e)$ are $O_p(\p\sqrt{\n})$.

The term (c) in \eqref{prop2:eq1} is 
\begin{eqnarray*}
&&\frac{1}{2\p} (c)
\cr
&&= \sum_{i=1}^\m \V_i^T \(\M^T\My -(1-\p) \diag(\M^T\My)\) \V_i
\cr
&&=\sum_{i=1}^\m \Bigg\{ 
	\sum_{k=1}^\n \( \sum_{h=1}^\d {\M}_{kh} \V_{ih} \)\( \sum_{h'=1}^\d (\y_{kh'}-\p) {\M}_{kh'} \V_{ih'} \)
\cr
&&\quad\quad\quad\quad\quad
	- (1-\p) \sum_{k=1}^\n \sum_{h=1}^\d (\y_{kh}-\p) {\M}_{kh}^2 \V_{ih}^2 \Bigg\}
\cr
&&=\sum_{k=1}^\n \sum_{h=1}^\d (\y_{kh}-\p) {\M}_{kh} \sum_{i=1}^\m \V_{ih} \Bigg[ \bigg(\sum_{h'=1}^\d {\M}_{kh'} \V_{ih'}\bigg) 
	- (1-\p){\M}_{kh} \V_{ih}\Bigg] .
\end{eqnarray*}
Then, its variance is 
\begin{eqnarray*}
&&\(\frac{1}{2\p}\)^2 \var(c) 
\cr
&&= \sum_{k=1}^\n \sum_{h=1}^\d \p (1-\p) {\M}_{kh}^2 \Bigg\{ \sum_{i=1}^\m \V_{ih} \Bigg[ \bigg(\sum_{h'=1}^\d {\M}_{kh'} \V_{ih'}\bigg) 
	- (1-\p){\M}_{kh} \V_{ih}\Bigg] \Bigg\}^2
\cr
&&\le C\p (1-\p)\n\d,
\end{eqnarray*}
where the last inequality is due to Assumption \ref{assume1}(1) and the fact that 
\begin{eqnarray} \label{prop2:keyOrder-M}
&&\sum_{k=1}^\n \sum_{h=1}^\d {\M}_{kh}^2 \Bigg\{ \sum_{i=1}^\m \V_{ih} \Bigg[ \bigg(\sum_{h'=1}^\d {\M}_{kh'} \V_{ih'}\bigg) 
	- (1-\p){\M}_{kh} \V_{ih}\Bigg] \Bigg\}^2
\cr
&&=\sum_{k=1}^\n \sum_{h=1}^\d {\M}_{kh}^2 \Bigg\{ \sum_{i=1}^\m \lam_i\U_{ik}\V_{ih} - (1-\p)\sum_{i=1}^\m{\M}_{kh} \V_{ih}^2 \Bigg\}^2
\cr
&&=\sum_{k=1}^\n \sum_{h=1}^\d {\M}_{kh}^2 \Bigg\{ \sum_{i=1}^\m \lam_i\U_{ik}\V_{ih} \Bigg\}^2
\cr
&&\quad\quad
	+ (1-\p)^2 \sum_{k=1}^\n \sum_{h=1}^\d \Bigg\{ \sum_{i=1}^\m{\M}_{kh}^2 \V_{ih}^2 \Bigg\}^2
\cr
&&\quad\quad
	-2(1-\p)\sum_{k=1}^\n \sum_{h=1}^\d {\M}_{kh} \Bigg\{ \sum_{i=1}^\m \lam_i\U_{ik}\V_{ih} \Bigg\}\Bigg\{ \sum_{i=1}^\m{\M}_{kh}^2 \V_{ih}^2 \Bigg\}
\cr
&&=\sum_{k=1}^\n \sum_{h=1}^\d {\M}_{kh}^2 \Bigg\{ \sum_{i=1}^\m \lam_i\U_{ik}\V_{ih} \Bigg\}^2
	+ O(\n)
\cr\cr
&&= O(\n\d).
\end{eqnarray}

The term (d) in \eqref{prop2:eq1} is 
\begin{eqnarray*}
\frac{1}{2\p} (d)
&=& \sum_{i=1}^\m \V_i^T \(\M^T\epsilony -(1-\p) \diag(\M^T\epsilony)\) \V_i
\cr
&=&\sum_{i=1}^\m \Bigg\{ 
	\sum_{k=1}^\n \bigg( \sum_{h=1}^\d {\M}_{kh} \V_{ih} \bigg)\bigg( \sum_{h'=1}^\d \y_{kh'}\epsilon_{kh'} \V_{ih'} \bigg)
\cr
&&\quad\quad\quad\quad\quad\quad\quad\quad
	- (1-\p) \sum_{k=1}^\n \sum_{h=1}^\d \y_{kh}\epsilon_{kh} {\M}_{kh} \V_{ih}^2 \Bigg\}
\cr
&=&\sum_{k=1}^\n \sum_{h=1}^\d \y_{kh}\epsilon_{kh} \Bigg\{ \sum_{i=1}^\m \V_{ih} \Bigg[ \bigg(\sum_{h'=1}^\d {\M}_{kh'} \V_{ih'}\bigg) - (1-\p){\M}_{kh} \V_{ih}\Bigg] \Bigg\} .
\end{eqnarray*}
Then, its variance is
\begin{align*}
\bigg(\frac{1}{2\p}\bigg)^2 \var(d) 
&= \sum_{k=1}^\n \sum_{h=1}^\d \p\sig^2 \Bigg\{ \sum_{i=1}^\m \V_{ih} \Bigg[ \bigg(\sum_{h'=1}^\d {\M}_{kh'} \V_{ih'}\bigg) - (1-\p){\M}_{kh} \V_{ih}\Bigg] \Bigg\}^2
\cr
&\le C\p \n\d ,
\end{align*}
where the last inequality is due to Assumption \ref{assume1}(1) and the fact that
\begin{eqnarray} \label{prop2:keyOrder}
&&\sum_{k=1}^\n \sum_{h=1}^\d \Bigg\{ \sum_{i=1}^\m \V_{ih} \Bigg[ \bigg(\sum_{h'=1}^\d {\M}_{kh'} \V_{ih'}\bigg) - (1-\p){\M}_{kh} \V_{ih}\Bigg] \Bigg\}^2
\cr
&&=\sum_{k=1}^\n \sum_{h=1}^\d \Bigg\{ \sum_{i=1}^\m \lam_i\U_{ik}\V_{ih} - (1-\p)\sum_{i=1}^\m{\M}_{kh} \V_{ih}^2 \Bigg\}^2
\cr
&&=  \sum_{i=1}^\m \lam_i^2 
	+ (1-\p)^2 \sum_{k=1}^\n \sum_{h=1}^\d \( \sum_{i=1}^\m{\M}_{kh} \V_{ih}^2 \)^2  
\cr
&&\quad\quad\quad\quad
	- 2(1-\p) \sum_{i=1}^\m \lam_i^2 \sum_{h=1}^\d \V_{ih}^2 \sum_{i'=1}^\m \V_{i'h}^2
\cr
&&= \sum_{i=1}^\m \lam_i^2 
	+O(\n).
\end{eqnarray}

The term (f) in \eqref{prop2:eq1} is 
\begin{eqnarray} \label{prop2:eq2}
|(f)|
&=&\left| \trace({\V_\p}^{(\m)T} \hat{\Sig}_\p {\V_\p}^{(\m)} - \O^T\V^{(\m)T} \hat{\Sig}_\p \V^{(\m)}\O ) \right|
\cr
&\le& \sum_{i=1}^m \left| \O_i^T\V^T \hat{\Sig}_\p \V\O_i - {\V_\p}_i^T \hat{\Sig}_\p {\V_\p}_i \right| \cr
&=& \sum_{i=1}^m \left\{ \left| (\V\O_i - {\V_\p}_i)^T \hat{\Sig}_\p (\V\O_i - {\V_\p}_i)  
	+ 2 {\lam_\p^2}_i {\V_\p}_i^T (\V\O_i - {\V_\p}_i) \right| \right\} \cr
&\le& \sum_{i=1}^m {\lam_\p^2}_1 \( \norm{\V\O_i - {\V_\p}_i}_2^2 + 2 \left| {\V_\p}_i^T (\V\O_i - {\V_\p}_i) \right| \) \cr
&=& \sum_{i=1}^m {\lam_\p^2}_1 \Big( \norm{\V\O_i - {\V_\p}_i}_2^2 \cr
&&\quad\quad\quad\quad\quad+ \left| \O_i^T\V^T \V\O_i - \O_i^T\V^T {\V_\p}_i - {\V_\p}_i^T \V\O_i + {\V_\p}_i^T {\V_\p}_i) \right| \Big) \cr
&=& \sum_{i=1}^m 2 {\lam_\p^2}_1 \norm{\V\O_i - {\V_\p}_i}_2^2 \cr
&=& 2 {\lam_\p^2}_1 \norm{\V^{(\m)}\O - {\V_\p}^{(\m)}}_F^2 \cr
&=& O_p(\p\d),
\end{eqnarray}
where $\O_i$ is the $i$-th column of $\O$ and 
	the last equality holds by Proposition \ref{prop1}, \eqref{sineDist}, and \eqref{prop3:eq2(i)-supp-supp}.

Therefore, the result follows from \eqref{prop2:eq1}-\eqref{prop2:eq2}.
\end{proof}

\begin{proof}[Proof of Proposition \ref{lem8}]
By Cram{\`e}r-Wold device, it is enough to show that for any given $(c_1, c_2)^T \in \real^2 \setminus (0,0)^T$,
\begin{eqnarray*}
&&\frac{1}{\sqrt{\n\d}\gamma_{c_1,c_2}}
\begin{pmatrix}
c_1 \smallskip \\ 
c_2
\end{pmatrix}
^T
\[
\begin{pmatrix}
\p^{-2} \sum_{i=1}^{\m} {\lam_\p^2}_i \smallskip \\ 
\p^2 \sum_{i=1}^{\m}( \lam^2_i + \n\sig^2 ) \,\phat
\end{pmatrix}
-
\begin{pmatrix}
\sum_{i=1}^{\m}\left[ \lam^2_i + \n\sig^2 \right] \smallskip \\ 
\p^3 \sum_{i=1}^{\m}( \lam^2_i + \n\sig^2 ) 
\end{pmatrix}
\]  
\cr\cr
&&
\rightarrow  \mathcal{N}\( 0, 1\) \text{ in distribution,} \quad\text{as } \n,\d\to\infty,
\end{eqnarray*}
where $\gamma_{c_1,c_2}^2 = \big(c_1 \, c_2\big) \Gamma_{\n\d}
\begin{pmatrix}
c_1 \smallskip \\ 
c_2
\end{pmatrix}
$.
When $c_1=0$, this can be directly showed by CLT. 
Thus, we only consider the case where $c_1\neq 0$.

We have
\begin{eqnarray} \label{lem8:linearCombo}
&&
\begin{pmatrix}
c_1 \smallskip \\ 
c_2
\end{pmatrix}
^T
\[
\begin{pmatrix}
\p^{-2} \sum_{i=1}^{\m} {\lam_\p^2}_i \smallskip \\ 
\p^2 \sum_{i=1}^{\m}( \lam^2_i + \n\sig^2 ) \,\phat
\end{pmatrix}
-
\begin{pmatrix}
\sum_{i=1}^{\m}\left[ \lam^2_i + \n\sig^2 \right] \smallskip \\ 
\p^3 \sum_{i=1}^{\m}( \lam^2_i + \n\sig^2 ) 
\end{pmatrix}
\]
\cr
&&=c_1\frac{1}{\p^2} \sum_{i=1}^{\m} \[ {\lam_\p^2}_i - \p^2\( \lam^2_i + \n\sig^2 \)\]
	+ c_2 \,\p^2 \sum_{i=1}^{\m}( \lam^2_i + \n\sig^2 )\(\phat-\p\)
\cr
&&= \frac{2c_1}{\p}\sum_{k=1}^\n \sum_{h=1}^\d (\y_{kh}-\p){\M}_{kh}  
	\sum_{i=1}^\m \V_{ih} \bigg[ \bigg(\sum_{h'=1}^\d {\M}_{kh'} \V_{ih'}\bigg) \cr
&&\hspace*{7cm} - (1-\p){\M}_{kh} \V_{ih}\bigg] \cr
&&\;\;\;+\frac{2c_1}{\p}\sum_{k=1}^\n \sum_{h=1}^\d \y_{kh}\epsilon_{kh}  
	\sum_{i=1}^\m \V_{ih} \bigg[ \bigg(\sum_{h'=1}^\d {\M}_{kh'} \V_{ih'}\bigg) - (1-\p){\M}_{kh} \V_{ih}\bigg] \cr
&&\;\;\;+o_p\(\sqrt{\frac{\n\d}{\p}}\)
	+\frac{c_2 \,\p^2}{\n\d} \sum_{k=1}^\n \sum_{h=1}^\d (\y_{kh}-\p)
	\sum_{i=1}^{\m}( \lam^2_i + \n\sig^2 )
\cr
&&= \sum_{k=1}^\n \sum_{h=1}^\d (\y_{kh}-\p) \bigg\{ 
	\frac{2c_1}{\p} {\M}_{kh} \sum_{i=1}^\m \V_{ih} \bigg[ \bigg(\sum_{h'=1}^\d {\M}_{kh'} \V_{ih'}\bigg)  
\cr
&&\quad\quad\quad\quad\quad\quad\quad\quad\quad\quad
	- (1-\p){\M}_{kh} \V_{ih}\bigg]+ \frac{c_2 \,\p^2}{\n\d} \sum_{i=1}^{\m}( \lam^2_i + \n\sig^2 ) \bigg\}
\cr
&&\;\;\;+\frac{2c_1}{\p}\sum_{k=1}^\n \sum_{h=1}^\d \y_{kh}\epsilon_{kh}  
	\sum_{i=1}^\m \V_{ih} \bigg[ \bigg(\sum_{h'=1}^\d {\M}_{kh'} \V_{ih'}\bigg) - (1-\p){\M}_{kh} \V_{ih}\bigg] \cr
&&\;\;\;+o_p\(\sqrt{\frac{\n\d}{\p}}\)
\cr
&&=(a) + (b) + o_p\(\sqrt{\frac{\n\d}{\p}}\),
\end{eqnarray}
where the second equality holds by Lemma \ref{prop2}.
Since the terms $(a)$ and $(b)$ are centered and 
	not correlated with each other under the model setup in Section \ref{setup}, we have
\begin{eqnarray} \label{lem8:var}
&&\var\[(a) + (b)\]=\var\[(a)\] + \var\[(b)\]
\cr
&&=\sum_{k=1}^\n \sum_{h=1}^\d \expect (\y_{kh}-\p)^2 \Bigg\{ 
	\frac{2c_1}{\p} {\M}_{kh} \sum_{i=1}^\m \V_{ih} \Bigg[ \bigg(\sum_{h'=1}^\d {\M}_{kh'} \V_{ih'}\bigg) 
\cr
&&\quad\quad\quad\quad\quad\quad\quad\quad\quad
	- (1-\p){\M}_{kh} \V_{ih}\Bigg] 
	+ \frac{c_2 \,\p^2}{\n\d} \sum_{i=1}^{\m}( \lam^2_i + \n\sig^2 ) \Bigg\}^2
\cr
&&\quad+\frac{4c_1^2}{\p^2}\,\sum_{k=1}^\n \sum_{h=1}^\d \expect\(\y_{kh}^2\epsilon_{kh}^2\)  
	\Bigg\{ \sum_{i=1}^\m \V_{ih} \Bigg[ \bigg(\sum_{h'=1}^\d {\M}_{kh'} \V_{ih'}\bigg) 
\cr
&&\hspace*{7cm}
	- (1-\p){\M}_{kh} \V_{ih}\Bigg] \Bigg\}^2
\cr
&&=\p(1-\p)\sum_{k=1}^\n \sum_{h=1}^\d \Bigg\{ 
	\frac{2c_1{\M}_{kh}}{\p} \sum_{i=1}^\m \lam_i\U_{ik}\V_{ih}
	+ c_2 \p^2 \sum_{i=1}^{\m} b^2_i \Bigg\}^2
\cr
&&\hspace*{0.5cm}
	+\frac{4\sig^2c_1^2}{\p}\,\sum_{i=1}^\m \lam_i^2 + O\(\frac{\n}{\p}\)
\cr
&&=\frac{4c_1^2(1-\p)}{\p}\sum_{k=1}^\n \sum_{h=1}^\d {\M}_{kh}^2\Bigg\{ 
	\sum_{i=1}^\m \lam_i\U_{ik}\V_{ih} \Bigg\}^2 
	+\frac{4\sig^2c_1^2}{\p}\,\sum_{i=1}^\m \lam_i^2 
\cr
&&\hspace*{0.5cm}
	+ 4c_1c_2 \,\n\d\,\p^2(1-\p) \Bigg(\sum_{i=1}^\m  b^2_i\Bigg)^2
	+ c_2^2 \,\n\d\, \p^5(1-\p) \(\sum_{i=1}^\m b_{i}^2\)^2
\cr
&&\hspace*{0.5cm}
	+ O\(\frac{\n}{\p}\)
\cr
&&=\n\d\, \big(c_1 \, c_2\big) \Gamma_{\n\d}
\begin{pmatrix}
c_1 \smallskip \\ 
c_2
\end{pmatrix}  + O\(\frac{\n}{\p}\),
\end{eqnarray}
where the third equality is due to \eqref{prop2:keyOrder-M}, \eqref{prop2:keyOrder} and Assumption \ref{assume1}(1).
Note that  
\begin{eqnarray} \label{lem8:var-order}
\n\d\, \big(c_1 \, c_2\big) \Gamma_{\n\d}
\begin{pmatrix}
c_1 \smallskip \\ 
c_2
\end{pmatrix}
\ge \frac{4c_1^2\sig^2}{\p} \sum_{i=1}^\m \lam_i^2 \ge \frac{c \,\n\d}{\p}.
\end{eqnarray}
Thus, Liapunov's condition is satisfied with $(a)+(b)$ because we have
\begin{eqnarray} \label{lem8:liapunov}
&&\sum_{k=1}^\n \sum_{h=1}^\d \expect \bigg| (\y_{kh}-\p) \bigg\{ 
	\frac{2c_1}{\p} {\M}_{kh} \sum_{i=1}^\m \V_{ih} \bigg[ \bigg(\sum_{h'=1}^\d {\M}_{kh'} \V_{ih'}\bigg) 
\cr
&&\quad\quad\quad\quad\quad\quad\quad\quad\quad\quad
	- (1-\p){\M}_{kh} \V_{ih}\bigg] + \frac{c_2 \,\p^2}{\n\d} \sum_{i=1}^{\m}( \lam^2_i + \n\sig^2 ) \bigg\}
\cr
&&\quad\;\;+ \y_{kh}\epsilon_{kh} \bigg\{ \frac{2c_1}{\p} 
	\sum_{i=1}^\m \V_{ih} \bigg[ \bigg(\sum_{h'=1}^\d {\M}_{kh'} \V_{ih'}\bigg) - (1-\p){\M}_{kh} \V_{ih}\bigg] \bigg\} \bigg|^3
\cr
&&\le 8 \sum_{k=1}^\n \sum_{h=1}^\d \Bigg\{ \expect | \y_{kh}-\p |^3 \bigg| 
	\frac{2c_1}{\p} {\M}_{kh} \sum_{i=1}^\m \V_{ih} \bigg[ \bigg(\sum_{h'=1}^\d {\M}_{kh'} \V_{ih'}\bigg) 
\cr
&&\hspace*{6cm}
	- (1-\p){\M}_{kh} \V_{ih} \bigg]
	+ O\(1\) \bigg|^3
\cr
&&\quad\quad\quad\quad\quad\quad\quad+ \expect |\y_{kh}\epsilon_{kh} |^3 \bigg| \frac{2c_1}{\p} 
	\sum_{i=1}^\m \V_{ih} \bigg[ \bigg(\sum_{h'=1}^\d {\M}_{kh'} \V_{ih'}\bigg) 
\cr
&&\hspace*{7cm}
	- (1-\p){\M}_{kh} \V_{ih}\bigg] \bigg|^3 \Bigg\}
\cr
&&\le \frac{C}{\p^2} \sum_{k=1}^\n \sum_{h=1}^\d \Bigg\{ 
	\bigg| \sum_{i=1}^\m \V_{ih} \bigg[ \bigg(\sum_{h'=1}^\d {\M}_{kh'} \V_{ih'}\bigg) 
\cr
&&\hspace*{5.5cm}
	- (1-\p){\M}_{kh} \V_{ih}\bigg] \bigg|^3 
	+ O\(1\) \Bigg\}
\cr
&&\le \frac{C}{\p^2} \sum_{k=1}^\n \sum_{h=1}^\d \Bigg\{ 
	\sum_{i=1}^\m |\V_{ih}|^3 \( \bigg|\sum_{h'=1}^\d {\M}_{kh'} \V_{ih'} \bigg|^3 + |\V_{ih}|^3 \) 
	+ O\(1\) \Bigg\}
\cr
&&\le \frac{C}{\p^2} \sum_{i=1}^\m\sum_{k=1}^\n \bigg|\sum_{h'=1}^\d {\M}_{kh'} \V_{ih'} \bigg|^3 + O\(\n\d\)
\cr
&&= O\Bigg(\frac{\n\d^{3/2}}{\p^2}\Bigg),
\end{eqnarray}
where the first inequality holds by Assumption \ref{assume1}(1),
	and the last two lines are due to Cauchy-Schwarz inequality.

By \eqref{lem8:linearCombo}-\eqref{lem8:liapunov}, Liapunov CLT and Slutsky theorem, we have
\begin{eqnarray*}
&&\frac{1}{\sqrt{\n\d}\gamma_{c_1,c_2}}
\begin{pmatrix}
c_1 \smallskip \\ 
c_2
\end{pmatrix}
^T
\[
\begin{pmatrix}
\p^{-2} \sum_{i=1}^{\m} {\lam_\p^2}_i \smallskip \\ 
\p^2 \sum_{i=1}^{\m}( \lam^2_i + \n\sig^2 ) \,\phat
\end{pmatrix}
-
\begin{pmatrix}
\sum_{i=1}^{\m}\left[ \lam^2_i + \n\sig^2 \right] \smallskip \\ 
\p^3 \sum_{i=1}^{\m}( \lam^2_i + \n\sig^2 ) 
\end{pmatrix}
\]
\cr\cr
&&\rightarrow \, \mathcal{N}\( 0, \, 1\) \text{ in distribution,} \quad\text{as }\n,\d\to\infty.
\end{eqnarray*}
\end{proof}

\begin{proof}[Proof of Proposition \ref{prop5}]
Similarly to the proof of \eqref{prop2:eq1}, we have
\begin{eqnarray*}
&&\hat{\tau}_{\p} - \n\p^2\sig^2 
\cr
&&= \frac{1}{\d-\rank} \trace\(\V_{\p c}^{T} \hat\Sig_\p \V_{\p c}\) 
		- \frac{1}{\d-\rank} \trace\(\Vc^{T} \expect \, \hat\Sig_\p \Vc\) \cr
&&= \frac{1}{\d-\rank} \trace \( \Vc^{T} \hat\Sig_\p \Vc \) 
\cr
&&\quad\quad
		+ \frac{1}{\d-\rank} \trace \( \V_{\p c}^{T} \hat\Sig_\p \V_{\p c} - \O^T\Vc^{T} \hat\Sig_\p \Vc\O^T \)
		- \frac{1}{\d-\rank} \trace\(\Vc^{T} \expect \, \hat\Sig_\p \Vc \) \cr
&&= \frac{1}{\d-\rank} \trace \Big(\Vc^{T}\big( \My^T\My -(1-\p) \diag(\My^T\My) 
\cr
&&\hspace*{6cm}
	- \p^2(1-\p)\,\diag(\M^T\M) \big)\Vc \Big) 
\cr
&&\quad\quad+ \frac{1}{\d-\rank} \trace \Big(\Vc^{T}\(\epsilony^T\epsilony -(1-\p) \diag(\epsilony^T\epsilony)\)\Vc - \n\p^2\sig^2 I_{\d-\rank}  \Big)
\cr
&&\quad\quad- 2\p(1-\p) \frac{1}{\d-\rank} \trace\Big(\Vc^{T} \(\diag(\My^T\M) \)\Vc \Big)
\cr
&&\quad\quad- 2\p(1-\p)\frac{1}{\d-\rank} \trace\Big(\Vc^{T}\(\diag(\epsilony^T\M)\)\Vc \Big)
\cr
&&\quad\quad+ 2\frac{1}{\d-\rank} \trace \Big(\Vc^{T}\(\My^T\epsilony -(1-\p) \diag(\My^T\epsilony)\)\Vc \Big) \cr
&&\quad\quad + \frac{1}{\d-\rank} \trace \( \V_{\p c}^{T} \hat\Sig_\p \V_{\p c} - \O^T\Vc^{T} \hat\Sig_\p \Vc\O \) \cr
&&= (A) + (B) - 2\p(1-\p)\cdot(C) - 2\p(1-\p)\cdot(D) + 2\cdot(E) + (F),
\end{eqnarray*}
where $\O \in \mathbb{V}_{\d-\rank,\d-\rank}$ is a solution to 
	$\inf_{\Q \in \mathbb{V}_{\d-\rank,\d-\rank}} \smallnorm{\V_{\p c}-\V_{c}\Q}_F^2$,
	and the third equality is due to the fact that $\M\Vc = \U\Lam\V^T\Vc = 0$. 
We will show that $(A)$-$(F)$ are $O_p\(\p\sqrt{\n}\)$.

Since the first five terms, $(A)$-$(E)$, are centered, 
	we only need to check their variances to find their rates. 
The variances of the terms $(A),(B),$ and $(E)$ are $O\(\p^2\n\)$, which can be shown similarly to the proof of \eqref{prop2:eq2(a)}.
The variance of the term $(C)$ is
\begin{eqnarray*}
\var (C) 
&\le& \frac{1}{\d-\rank} \sum_{i=1}^{\d-\rank} \expect\left[ {\Vc}_i^{T} \(\diag(\My^T\M) \) {\Vc}_i \right]^2 
\cr
&=& \frac{1}{\d-\rank} \sum_{i=1}^{\d-\rank} \var\Bigg[ 
	\sum_{k=1}^\n \sum_{h=1}^\d (\y_{kh}-\p){\M}_{kh}^2 {\Vc}_{ih}^2 \Bigg] 
\cr
&=& \frac{1}{\d-\rank} \sum_{i=1}^{\d-\rank} \Bigg[ 
	\L^4\,\sum_{k=1}^\n O(\p(1-\p)) \Bigg]
\cr
&=& O(\p\n),
\end{eqnarray*}
where the inequality is due to Jensen's inequality. 
Similarly, the variance of the term $(D)$ is $O(\p\n)$.

Now, consider the term $(F)$. Similarly to the proof of \eqref{prop2:eq2},
\begin{eqnarray*}
|(F)|
&\le&\frac{1}{\d-\rank} 
	\left| \trace \( \V_{\p c}^{T} \hat\Sig_\p \V_{\p c} - \O^T\Vc^{T} \hat\Sig_\p \Vc\O \) \right|
\cr
&\le& \frac{1}{\d-\rank} \sum_{i=1}^{\d-\rank} \left| {\V_{\p c}}_i^{T} \hat\Sig_\p {\V_{\p c}}_i - \O_i^T\Vc^{T} \hat\Sig_\p \Vc\O_i  \right| \cr
&\le& \frac{1}{\d-\rank} \cdot 2 {\lam_\p^2}_1 \norm{ \V_{\p c} - \Vc\O }_F^2 \cr
&\le& \frac{1}{\d-\rank} \cdot 4 {\lam_\p^2}_1 \norm{ \sin\( \V_{\p c}, \Vc\) }_F^2 \cr
&=& \frac{1}{\d-\rank} \cdot 2 {\lam_\p^2}_1 \norm{ \V_{\p c} \V_{\p c}^T - \Vc \Vc^T }_F^2 \cr
&=& \frac{1}{\d-\rank} \cdot 2 {\lam_\p^2}_1 \norm{ (I_\d - \V_\p \V_\p^T) - (I_\d - \V \V^T ) }_F^2 \cr
&=& \frac{1}{\d-\rank} \cdot 2 {\lam_\p^2}_1 \norm{ \V_\p \V_\p^T - \V \V^T }_F^2 \cr
&=& \frac{1}{\d-\rank} \cdot 4 {\lam_\p^2}_1 \norm{ \sin (\V_\p, \V) }_F^2 \cr
&=& O_p(\p),
\end{eqnarray*}
where $\O_i$ is the $i$-th column of $\O$,
	the third inequality can be derived similarly to the proof of \eqref{prop3:eq2(i)-supp}, and 
	the last equality holds by Proposition \ref{prop1} and \eqref{prop3:eq2(i)-supp-supp}.
\end{proof}

\subsection{Proofs for Section \ref{proofs:thm3}} \label{apdx3}

\begin{proof}[Proof of Proposition \ref{prop4}]
Let $\Delta_{\lam_i} = \lamhat_i-\lam_i$, $\Delta_{\U_i} = \text{sign}(\langle \Uhat_i, \U_i \rangle)\Uhat_i-\U_i$, and $\Delta_{\V_i} = \text{sign}(\langle \Vhat_i, \V_i \rangle)\Vhat_i-\V_i$ for all $i \in \{1, \ldots,\rank \}$. 
Similarly to the proof of Theorem \ref{thm2}, we can show that for all $i=1,\ldots,\rank$, 
\begin{eqnarray} \label{prop4:lamDelta}
\left|\Delta_{\lam_i}\right| = O_p \(\frac{1}{\sqrt{\p}}+\frac{1}{\p}\sqrt{\frac{\d}{\n}}\). 
\end{eqnarray}
Then,
\begin{eqnarray*}
&&\norm{\Mhat(\s_0)-\M}_F^2
\cr
&&= \norm{\sum_{i=1}^\rank \s_{0i} \; \lamhat_i \Uhat_i \Vhat_i^T - \sum_{i=1}^\rank \lam_i \U_i \V_i^T}_F^2
\cr
&&\le \rank^2 \sum_{i=1}^\rank \norm{ \s_{0i} \; \lamhat_i \Uhat_i \Vhat_i^T - \lam_i \U_i \V_i^T}_F^2
\cr
&&= \rank^2 \sum_{i=1}^\rank \norm{ \( \lam_i + \Delta_{\lam_i} \) \( \U_i + \Delta_{\U_i} \) \( \V_i + \Delta_{\V_i} \)^T - \lam_i \U_i \V_i^T}_F^2
\cr
&&\le C \rank^2 \sum_{i=1}^\rank \bigg\{ \Big\| \Delta_{\lam_i}\U_i \V_i^T\Big\|_F^2 + \Big\| \lam_i \Delta_{\U_i} \V_i^T\Big\|_F^2 + \Big\|\lam_i \U_i \Delta_{\V_i}^T\Big\|_F^2 \bigg\}
\cr
&&= C\rank^2 \sum_{i=1}^\rank \Bigg\{ O_p \(\frac{1}{\sqrt{\p}}+\frac{1}{\p}\sqrt{\frac{\d}{\n}}\) 
	+ O\(\n\d\) \frac{1}{\p\,{b}_{\rank}^4}O_p\(\frac{1}{\d}\) 
	+ O\(\n\d\)\frac{1}{\p\,{b}_{\rank}^4} O_p\(\frac{1}{\n}\) \Bigg\}
\cr
&&=\frac{1}{\p\,{b}_{\rank}^4} \, O_p(\n),
\end{eqnarray*}
where the third equality holds due to \eqref{prop4:lamDelta} and Theorem \ref{thm1}.
\end{proof}

\subsection{Proofs for Lemma \ref{thm5}} \label{apdx4}


\begin{proof}[Proof of Lemma \ref{thm5}]
By Weyl's theorem (\cite{li1998one}), Lemma \ref{lem6}, and Lemma \ref{lem4}, for any given $\delta>0$, 
	there exists a large constant $C_{\delta}>0$ such that
\begin{eqnarray}\label{thm5:eq1}
\max\left\{\big|\lam_{\phat \rank}^2 -\p^2(\lam_\rank^2 + \n\sig^2)\big|,\big|\lam_{\phat \,\rank+1}^2 - \p^2\n\sig^2\big|\right\}
&\le& \norm{\hat\Sig_{\phat}-\expect(\hat\Sig_{\p})}_2\cr
	&\le& \norm{\hat\Sig_{\phat}-\hat\Sig_{\p}}_2+\norm{\hat\Sig_{\p}-\expect(\hat\Sig_{\p})}_2\cr
	&\le& C_{\delta}\,\p^{3/2}\sqrt{\frac{\n\log\n}{\d}}
\end{eqnarray}
with probability at least $1-O(\n^{-\delta})$. 
Also, by definition of $\hat\rank$, we have 
\begin{eqnarray}\label{thm5:eq2}
\big\{\hat\rank = \rank\big\}&=&\big\{\lam_{\phat \rank}^2 \ge \p^2\n\, C_\d ,\; \lam_{\phat \,\rank+1}^2 < \p^2\n\, C_\d \big\}\cr
	&=&\Big\{\big[\lam_{\phat \rank}^2 -\p^2(\lam_\rank^2 + \n\sig^2)\big] + \p^2(\lam_\rank^2 + \n\sig^2) \ge \p^2\n\, C_\d ,\; \cr
	&&\quad\quad\quad\quad\quad\quad\quad 
		\big[\lam_{\phat \,\rank+1}^2 - \p^2\n\sig^2\big] + \p^2\n\sig^2 < \p^2\n\, C_\d \Big\},
\end{eqnarray}
where $\lam_\rank^2=b_r^2\, \n\d$ by Assumption \ref{assume1}(1).
The result follows by \eqref{thm5:eq1} and \eqref{thm5:eq2}.

\end{proof}



\markboth{\hfill{\footnotesize\rm Juhee Cho, Donggyu Kim, and Karl Rohe} \hfill}
{\hfill {\footnotesize\rm Asymptotic theory for LRMC} \hfill}

\bibhang=1.7pc
\bibsep=2pt
\fontsize{9}{14pt plus.8pt minus .6pt}\selectfont
\renewcommand\bibname
\vskip .65cm
\vskip 2pt
\noindent
Department of Statistics, University of Wisconsin, Madison, WI 53706, U.S.A
\vskip 2pt
\noindent
E-mail: chojuhee@stat.wisc.edu
\vskip 2pt

\noindent
Department of Statistics, University of Wisconsin, Madison, WI 53706, U.S.A
\vskip 2pt
\noindent
E-mail: kimd@stat.wisc.edu
\vskip 2pt

\noindent
Department of Statistics, University of Wisconsin, Madison, WI 53706, U.S.A
\vskip 2pt
\noindent
E-mail: karlrohe@stat.wisc.edu
\end{document}